\documentclass[sigconf]{acmart}
\renewcommand\footnotetextcopyrightpermission[1]{}

\usepackage[subtle]{savetrees}
\usepackage{amsmath,amsthm}
\usepackage{graphicx}
\usepackage{textcomp}
\usepackage{xcolor}
\usepackage[linesnumbered,ruled,noend]{algorithm2e}
\usepackage{array}
\usepackage{caption}
\usepackage{colortbl}
\usepackage{endnotes}
\usepackage[utf8]{inputenx}
\usepackage{microtype}
\usepackage{multirow}
\usepackage{setspace}
\usepackage{subcaption}
\usepackage{tabularx}
\usepackage{todonotes}
\usepackage{bm}
\usepackage{xspace}
\usepackage{mathtools}
\usepackage{comment}

\usepackage[inline]{enumitem}
\usepackage{thmtools}
\usepackage{thm-restate}
\usepackage[binary-units=true]{siunitx}
\usepackage[font=footnotesize]{caption}

\definecolor{mygray}{RGB}{90,90,90}

\allowdisplaybreaks

\def\BibTeX{{\rm B\kern-.05em{\sc i\kern-.025em b}\kern-.08em
		T\kern-.1667em\lower.7ex\hbox{E}\kern-.125emX}}
\newcolumntype{Y}{>{\centering\arraybackslash}X}
\let\en=\ensuremath
\DeclareRobustCommand{\stirling}{\genfrac\{\}{0pt}{}}

\renewcommand{\vec}[1]{\en{\bm{\mathrm{#1}}}}
\newcommand{\mat}[1]{\vec{#1}}

\newcommand\notsoscriptsize{\@setfontsize\notsoscriptsize\@vipt\@viipt}

\DeclareMathOperator{\E}{\mathbb{E}}            %
\DeclareMathOperator{\Var}{\mathrm{Var}}        %
\DeclarePairedDelimiter{\abs}{\lvert}{\rvert}   %
\DeclarePairedDelimiter{\ip}{\langle}{\rangle}  %
\DeclarePairedDelimiter{\norm}{\lVert}{\rVert}  %

\renewcommand{\th}[0]{\textsuperscript{th}\xspace}

\SetKw{Continue}{continue}

\AtBeginDocument{%
	\providecommand\BibTeX{{%
			\normalfont B\kern-0.5em{\scshape i\kern-0.25em b}\kern-0.8em\TeX}}}

\begin{document}

	\title{Building a Collaborative Phone Blacklisting System with Local Differential Privacy}

	\author{Daniele Ucci}
	\affiliation{%
		\institution{Department of Computer, Control, and Management Engineering ``Antonio Ruberti\rq\rq, ``La Sapienza\rq\rq\ University of\ Rome}
		\city{Rome}
		\state{Italy}
	}
	\email{ucci@diag.uniroma1.it}
	
	\author{Roberto Perdisci}
	\affiliation{%
		\institution{University of Georgia}
		\city{Atlanta}
		\country{Georgia}}
	\email{perdisci@cs.uga.edu}
	
	\author{Jaewoo Lee}
	\affiliation{%
		\institution{University of Georgia}
		\city{Atlanta}
		\country{Georgia}}
	\email{jaewoo.lee@cs.uga.edu}
	
	\author{Mustaque Ahamad}
	\affiliation{%
		\institution{Georgia Institute of Technology}
		\city{Atlanta}
		\country{Georgia}}
	\email{mustaq@cc.gatech.edu}

	\begin{abstract}
Spam phone calls have been rapidly growing from nuisance to an increasingly effective scam delivery tool. To counter this increasingly successful attack vector, a number of commercial smartphone apps that promise to block spam phone calls have appeared on app stores, and are now used by hundreds of thousands or even millions of users. However, following a business model similar to some online social network services, these apps often collect call records or other potentially sensitive information from users' phones with little or no formal privacy guarantees.

In this paper, we study whether it is possible to build a practical collaborative phone blacklisting system that makes use of local differential privacy (LDP) mechanisms to provide clear privacy guarantees. We analyze the challenges and trade-offs related to using LDP, evaluate our LDP-based system on real-world user-reported call records collected by the FTC, and show that it is possible to learn a phone blacklist using a reasonable overall privacy budget and at the same time preserve users' privacy while maintaining utility for the learned blacklist.
\end{abstract}


	\maketitle
	\pagestyle{plain}

\section{Introduction}
\label{sec:intro}

Spam phone calls have been rapidly growing from nuisance to supporting well-coordinating fraudulent campaigns~\cite{NYT:robocalls,IRS:scams,WaPo:robocalls}. To counter this increasingly successful attack vector, federal agencies such as the US Federal Trade Commission (FTC) have been working with telephone carriers to design systems for blocking robocalls (i.e., automated calls)~\cite{FTC:PolicyForum,FTC:RobocallDefense}. At the same time, a number of smartphone apps that promise to block spam phone calls have appeared on app stores~\cite{YouMail,TrueCaller,Nomorobo}, and smartphone vendors, such as Google~\cite{Google:PhoneApp}, are embedding some spam blocking functionalities into their default phone apps.

Currently, most spam blocking apps rely on caller ID blacklisting, whereby calls from phone numbers that are known to have been involved in spamming or numerous unwanted calls are blocked (either automatically, or upon explicit user consent). Recently, Pandit et al.~\cite{Pandit2018} have studied how to learn such blacklists from a variety of data sources, including user complaints, phone honeypot call detail records (CDRs) and call recordings. Existing commercial apps, such as Youmail~\cite{YouMail} and  TouchPal~\cite{TouchPal}, mostly base their blocking approach on user complaints. Other popular apps, such as TrueCaller~\cite{TrueCaller}, also use information collected from users' contact lists to distinguish between possible legitimate and unknown/unwanted calls\footnote{These behavior are inferred merely from publicly available information; further details on the inner-workings of commercial apps are difficult to obtain and their technical approach cannot be fully evaluated for comparison.}. However, in the recent past TrueCaller has experienced significant backlash due to privacy concerns related to the sharing of users' contact lists with a centralized service.
Google recently implemented a built-in feature in Android phones to protect against possible spam calls. Nonetheless, Android phones may send information about received calls to Google without strong privacy guarantees~\cite{Google:PhoneApp}.

While learning a blacklist from CDRs collected by phone honeypots~\cite{Pandit2018} is a promising approach that poses little or no privacy risks, it suffers from some drawbacks. First, operating a phone honeypot is expensive, as thousands of phone numbers have to be acquired from telephone carriers. Furthermore, in~\cite{Pandit2018} it has been reported that the spam calls targeting the honeypot were skewed towards business-oriented campaigns, likely because the honeypot numbers were mostly re-purposed business numbers (perhaps because re-purposing a user's number may pose some privacy risks, since others might still try to reach a specific person at that number). Conversely, leveraging user complaints also has some drawbacks. For instance, for a user to be able to complain or label a number (as in TouchPal~\cite{TouchPal}), the user has to answer to and identify the purpose of the call. However, only a fraction of users typically answers and listens to calls from unknown numbers (i.e., numbers not registered in the contact list). Furthermore, user-provided call labels are quite noisy and a relatively high number of complaints about the same number need to be observed, before being able to accurately label the source of the calls~\cite{TouchPal}. This may delay the insertion of a spam number into the blacklist, thus leaving open a time window for the spammers to succeed in their campaigns.

One possible solution would be to use an approach similar to the CDR-based blacklisting proposed in~\cite{Pandit2018}, while using real phone numbers as ``live honeypots.'' In other words, if a smartphone app could leverage the call logs of real phone users without requiring the users to explicitly label the phone calls, this would provide a solution to the drawbacks mentioned above. Unfortunately, this may obviously pose serious privacy risks to users. For instance, knowing that a user received a phone call from a specific phone number related to a cancer treatment clinic may reveal that the user (or a close family member) is a cancer patient.

{\bf Research Question}: {\em Can these privacy concerns be mitigated, and the users' call logs be collaboratively contributed to enable learning an accurate phone blacklist with strong privacy guarantees?}

To answer the above research question, in this paper we study whether it is
possible to design a {\em practical} phone blacklisting system that leverages
differential privacy~\cite{Dwork06differentialprivacy} mechanisms to
collaboratively learn effective anti-spam phone blacklists while providing
strong privacy guarantees. Specifically, we leverage a state-of-the art {\em
local differential privacy} (LDP) mechanisms for {\em generic heavy-hitter
detection} that has been shown to work only in theory~\cite{Bassily2015}, and
focus on adapting it to enable the implementation of a concrete
privacy-preserving collaborative blacklist learning system that could be
deployed on real smartphone devices. To the best of our knowledge, we are the
first to study the application of local differential privacy to building
blacklist-based defenses, and specifically towards defenses against telephony
spam.

\begin{figure}
\centering
\includegraphics[scale=0.36]{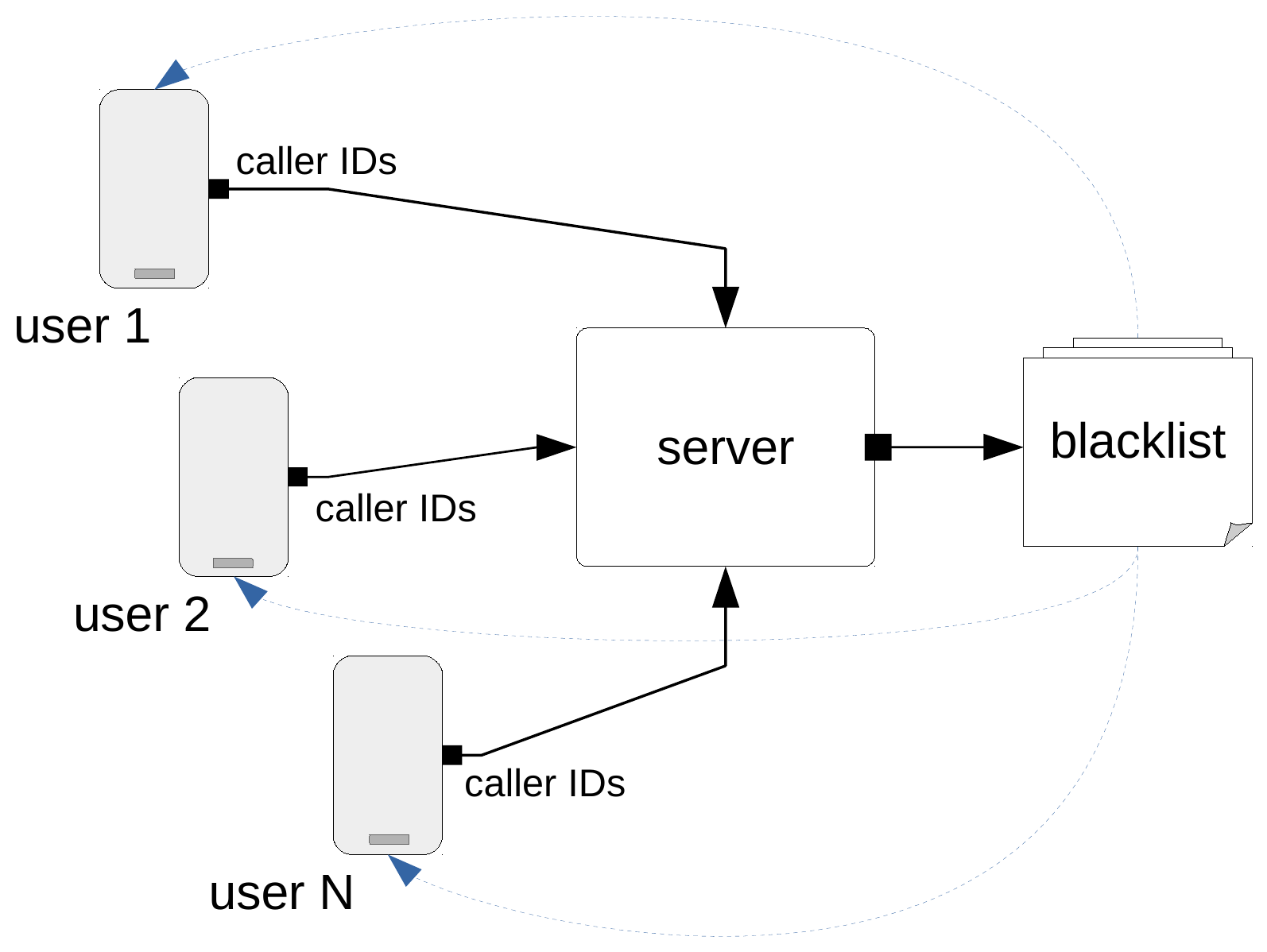}
\caption{System overview. Caller IDs are collected with local differential privacy. After learning, blacklist updates are propagated back to users.}
\label{fig:sys_overview}
\end{figure}

Figure~\ref{fig:sys_overview} shows an overview of our system. Participating users install an app that can implement the following high-level functionalities (more details about the client app are provided in Section~\ref{sec:evaluation}): when the user receives a call, the app will first check if the caller ID (i.e., the calling number) is in the users' contacts list; if yes, the caller ID is considered as {\em trusted} and ignored, otherwise the caller ID is considered to be {\em unknown} and buffered for reporting. Unknown caller IDs are then checked against a blacklist; if a match is found, the user can be alerted that the incoming phone call originates from a phone number known to have been involved in spamming activities, so that the user can decide whether to reject the call. At the end of a predefined time window (e.g., once daily), the app will report unknown caller IDs from which the user received a phone call (including both unanswered and accepted calls). Consequently, the server will receive daily reports from each user, which consist of the list of unknown caller IDs observed by the client apps running on each device. As we will explain in Section~\ref{sec:evaluation}, all the caller IDs are delivered by the client apps to the server via a novel LDP mechanism. This is done to provide privacy guarantees and minimize the risk of the server learning any sensitive information about single users' phone calls (e.g., whether the user may be a cancer patient, given that she has received calls from a cancer treatment clinic). At the same time, while the users' privacy is protected, the server is able to identify {\em heavy hitter} caller IDs that are highly likely associated with new spamming activities.
Hence, our system preserves user privacy by making it difficult for the server to learn the list of caller IDs that are contacting the users, while keeping its capability of building a blacklist of possible spammers.

While LDP mechanisms provide strong privacy guarantees, they are often studied in
theoretical terms and their applicability to practical, real-world security
problems is often left as a secondary consideration. On the contrary, in this
paper we focus primarily on adapting a state-of-the-art LDP mechanism for heavy
hitter detection~\cite{Bassily2015} to make it {\em practical}, so that it can
be used in the smartphone app described above to report the list of caller IDs
to the server. Furthermore, we evaluate the ability of the server to accurately
reconstruct the (noisy) reported caller IDs under different privacy budgets, and
evaluate the utility of the learned blacklist. To this end, we implement both
the client-side (i.e., smartphone side) and server-side (i.e., blacklist
learning side) LDP protocol, leaving other app implementation details (e.g.,
user preferences and controls) to future work.

In summary, we make the following contributions:
\begin{itemize}
\item We explore how to build a privacy-preserving collaborative phone blacklisting system using local differential privacy (LDP). Specifically, we expose what are the challenges related to building a {\em practical} LDP-based system that is able to learn a phone blacklist from caller ID data provided by a pool of contributing users, and propose a number of approaches to overcome these challenges. To the best of our knowledge, our system is the first application of LDP protocols to building a defense against phone spam.

\item We implement our blacklisting system using a new LDP protocol for heavy
hitter detection. Our protocol is built upon a state-of-the-art protocol
previously proposed in~\cite{Bassily2015}. We first show that~\cite{Bassily2015}
is not practical, in that it cannot be applied {\em as is} to collaborative
phone blacklisting. We then introduce novel LDP protocol modifications, such as
data bucketization and variance-reduction mechanisms, to enable heavy hitter
detection by building a LDP-based phone blacklisting approach that could be
deployed on real smartphones.

\item We evaluate our LDP-based system on real-world user reported call records collected by the FTC. Specifically, we analyze multiple different trade-offs, including the trade-off between the privacy budget assigned to the different components of our LDP protocol and the overall blacklist learning accuracy. Our results indicate that it is possible to learn a phone blacklist using a reasonable overall privacy budget, and to preserve users' privacy while maintaining utility for the learned blacklist.

\end{itemize}

\section{Problem Definition and Approach}
\label{sec:problem_def}

In this section, we outline our threat model and briefly describe our approach towards collaboratively building phone blacklists in a privacy-preserving way.

\vspace{3pt}
\noindent
{\bf Threat Model}~~
In designing our phone blacklisting system (see Figure~\ref{fig:sys_overview}), we make the following assumptions:
\begin{itemize}

\item We consider the caller ID related to phone calls received by users as privacy sensitive (e.g., see the cancer clinic example given in Section~\ref{sec:intro}). However, we do not consider the caller ID area code prefix (e.g., the first three digit of a US phone number) as sensitive. The reason is that each area code includes millions of possible phone numbers (e.g., $10^7$ numbers in the US). Therefore, even if the attacker learns that a given user received a phone call from a given area code prefix, she would be faced with very high uncertainty regarding what specific number actually called the user.

\item We assume the privacy-preserving data collection app running on each user's device is trusted. Namely, we assume the app correctly implements our proposed LDP protocol (detailed in Algorithm~\ref{alg:client_side_variant}), and that it does not directly collect and report any other user data to the server other than the {\em unknown} phone numbers from which calls were received.

\item We also assume that the server correctly executes the server-side of our LDP protocol, to learn a useful phone blacklist that can be propagated back to the users to help them block future spam calls. At the same time, we assume that the server may at some point be compromised (or subpoenaed), allowing an adversary to access future users' reports. Unlike traditional curator-based differential privacy mechanisms, our use of LDP mechanisms guarantees that, in the event of a breach of the server, the privacy of users' phone call records can still be preserved (see Section~\ref{sec:background}, for details).

\end{itemize}

It is worth noting that the server may be able to observe the IP address of each reporting device. Furthermore, in a practical deployment, the server may realistically implement an authentication mechanism that requires users to register to the blacklisting service (e.g., by providing an email address, password, etc.), to be allowed to (privately) report call records and receive blacklist updates. In this case, the identity of the users may be known to the server, and a server breach may expose such identities. However, in this paper {\em we focus exclusively on protecting the privacy of users' phone call records}, rather than anonymity. Protecting the IP address and identity of users may be achieved via other security mechanisms that are outside the scope of this work.

\vspace{3pt}
\noindent
{\bf Approach Overview}~~
According to recent work on phone blacklisting~\cite{Pandit2018,Liu2018}, it is clear that most spammers will tend to call a large number of users, in an attempt to identify a subset of them who may fall for a scam. Therefore, given a large and distributed user population, it is reasonable to consider {\em heavy hitters} as candidate spammers. In other words, a caller ID that is reported as {\em unknown} by a significant fraction of participating smartphones satisfies the volume and diversity features used in previous work~\cite{Pandit2018,Liu2018}, and can be considered for blacklisting.

Following the high-level approach proposed in previous work, we therefore cast the problem of learning a phone blacklist as a {\em heavy hitter detection} problem. The main research question we investigate in this paper is the following: using the system depicted in Figure~\ref{fig:sys_overview}, is it possible to accurately detect heavy hitter caller IDs while providing local differential privacy guarantees?

To investigate the above research question, we start from a state-of-the-art LDP protocol for heavy hitter detection proposed by Bassily and Smith~\cite{Bassily2015}, which throughout the rest of the paper we will refer to as SH (short for {\em succinct histogram}). Unfortunately, we have found that the SH protocol is not suitable {\em as is} for providing a solution to our application scenario (explained in details in Section~\ref{sec:system_and_attacker_models}). Among the main issues we found is the fact that SH tends to work well only {\em in expectation}. As we aim to build a practical blacklisting system, we would like our system to perform well for realistic, limited population sizes (e.g., thousands of users). Furthermore, the protocol used in~\cite{Bassily2015} for calculating the frequency of occurrence for a heavy hitter (i.e., the number of calls made by a likely spammer, in our case) is complex and difficult to implement efficiently (to the best of our knowledge, no implementation of the full~\cite{Bassily2015} protocol is publicly available).

To address the above limitations of the SH protocol, we introduce three LDP protocol modifications:
\begin{enumerate}
\item We propose a novel response randomizer that has the effect of reducing the
variance in the noisy inputs received by the server-side of the SH protocol,
thus increasing server-side heavy hitter reconstruction accuracy even in the
case of a limited user population (Section~\ref{sec:system_and_attacker_models}).
\item Second, we replace the frequency oracle part of the SH protocol proposed in~\cite{Bassily2015} with a much simpler protocol recently proposed in~\cite{OLH}, whose implementation is publicly available (Section~\ref{sec:OLH}).
\item To increase the relative frequency of heavy hitter caller IDs and boost the likelihood that the server will be able to correctly reconstruct them and add them to the blacklist, we introduce a bucketization mechanism. In essence, before a user (more precisely, the app running on the user's phone) reports one or more caller IDs to the server, the caller IDs are first grouped according to their three-digit area code. Then, the client-side portion of the SH protocol is run independently per each single group (i.e., per each area code). The intuition here is that some spammers tend to use phone numbers from specific area codes. For instance, IRS phone scams are often performed using caller IDs with a 202 prefix (Washington DC area code), as this may trick more users into believing it is truly the IRS that is calling. By grouping caller IDs based on area code, spam numbers also tend to group, increasing their relative frequency compared to all other caller IDs with the same prefix. This effect is discussed in details in Section~\ref{sec:ac_buckets}.
\end{enumerate}

Section~\ref{sec:system_and_attacker_models} presents the details of our LDP protocol.

\vspace{3pt}
\noindent
{\bf Caller ID Spoofing}~~
Caller ID spoofing is the main limiting factor for the effectiveness of phone
blacklists in general, as also acknowledged in previous
work~\cite{Pandit2018,TouchPal}. Previous research on phone
blacklisting~~\cite{Pandit2018,TouchPal} regards the prevention of caller ID
spoofing as an orthogonal research direction, leaving it to future work. This
choice can be justified by noting that the FCC has mandated that all US phone
companies must implement caller ID authentication by June 30, 2021~\cite{FCCmandate}. In response,
telephone carriers have started activating an authentication protocol known as
SHAKEN/STIR~\cite{WaPo:robocalls}. In our work, we make similar considerations
as in previous work, and focus our attention on the feasibility of building
phone blacklists using user-provided data with strong privacy guarantees. We
therefore consider dealing with caller ID spoofing to be outside the scope of
this paper.

\section{Background}
\label{sec:background}

\subsection{Notation}
Suppose there are $n$ users, and that each user $j$ holds
an item $v_j$ drawn from a domain $\mathcal{V}$ of size $d$
(in our case, $\mathcal{V}$ is the set of valid phone numbers).
For each item $v \in \mathcal{V}$, its frequency
$f(v)$ is defined as the fraction of users who hold $v$, i.e.,
$f(v) = \abs{\{j \in [n]: v_j = v\}}/n$, where $[n]$ denotes the set $\{1,
2,\ldots, n\}$. For notational simplicity, we omit the subscript $j$
when it's clear from the context.

A frequency oracle (FO) is a function that can (privately)
estimate the frequency of any item $v \in \mathcal{V}$ among the user population.

For a vector $\vec{x}=(x_1, \ldots, x_m)$, we will use the
array index notation $\vec{x}[i]$ to denote the
$i^{\mathrm{th}}$ entry, i.e., $\vec{x}[i] = x_i$. Similarly,
$\mat{X}[i, j]$ denotes the entry at location $(i, j)$ for a
matrix $\mat{X}$.

\subsection{Local Differential Privacy}
\label{sec:LDP}
Differential privacy can be applied to two different settings:
centralized and local. In the centralized setting, it is assumed that
there exists a {\em trusted data curator} who collects personal data
$\vec{v}=(v_1, \ldots, v_n)$ from users, analyzes it, and releases the
results after applying a differentially private transformation.
On the other hand, in the local setting there is {\em no single
trusted third party}.
To protect privacy, each user independently perturbs her record $v_j$
into $\tilde{v}_j=\mathcal{A}(v_j)$ using a
randomized algorithm $\mathcal{A}$, and only shares the perturbed
version with an aggregator (the centralized server responsible for blacklist
learning, in our application).
The local differential privacy (LDP) model provides stronger
privacy protection than the centralized model, because it protects privacy
even when the aggregator (i.e., the blacklist learning server, in our case)
is compromised and controlled by an adversary.
The level of privacy protection
depends on a privacy budget parameter $\varepsilon$, as formally defined in~\cite{Duchi2013};
the smaller $\varepsilon$, the greater the privacy guarantees.

\subsection{The Succinct Histogram Protocol}
\label{sec:BS_proto}

Bassily and Smith~\cite{Bassily2015} proposed an $\varepsilon$-LDP
protocol, called Succinct Histogram (SH), for detecting heavy hitters
over a large domain $\mathcal{V}$.
In their work, the authors assume that each user has a single item to share with
the server.

Unfortunately, in~\cite{Bassily2015} the client- and server-side of the protocol are presented
as ``interleaved'' in a single algorithm, and to the best of our knowledge a
practical implementation of the client-server protocol was not provided. 
To make
the LDP protocol in~\cite{Bassily2015} practical and applicable to our
collaborative blacklist learning system, we provide a new but equivalent representation of the protocol
proposed in~\cite{Bassily2015}
that focuses on the interactions between clients
(i.e., the system contributors) and server. Due to space limitations, we report
our new client-server formulation in Appendix~\ref{apx:background} (see
Algorithms~\ref{alg:client_side} and~\ref{alg:server_side}).

The SH protocol works as follows.  First, %
each user $j \in [n]$ encodes her item $v_j \in \mathcal{V}$
into a bit string of length $m$ using a binary error-correcting code
$(\mathsf{Enc}, \mathsf{Dec})$\footnote{Specifically, the protocol requires a $\left[2^m,k,d\right]_2$ binary error-correcting code, where $2^m$, $k$, and $d$ represent the codeword length, encoded message length (in bits), and minimum distance, respectively, in which the relative distance $d/2^m$ has to be included in the interval $\left(0,1/2\right)$.
}.
For notational simplicity, we let
$\mathsf{Enc}(\cdot) = \vec{c}(\cdot)$. Let $\vec{x}_j \in \{-1/\sqrt{m},
1/\sqrt{m}\}^m$ be the encoded binary string. The encoded item
$\vec{x}_j$ and its decoding are respectively given by
\[
\vec{x}_j =\mathsf{Enc}(v_j) = \vec{c}(v_j) \text{ and }
\mathsf{Dec}(\vec{x}_j) = v_j\,.
\]
For privacy, each user $j$ perturbs $\vec{x}_j$ %
into a noisy report $\vec{z}_j = \mathcal{R}_{\mathsf{bas}}(\vec{x}_j,
\varepsilon)$ using a randomizer $\mathcal{R}_{\mathsf{bas}}$ and
sends it to the server.
The pseudo-code of randomizer $\mathcal{R}_{\mathsf{bas}}$ is described in
Algorithm~\ref{alg:basic_rand}.

To simplify the heavy hitter detection problem, Bassily and Smith
applied the idea of isolating heavy hitters into different channels
using a pairwise independent hash function $H:\mathcal{V} \to [K]$,
whereby an item $v$ is mapped to channel $H(v)$.

This has the effect that, with high probability,
no two unique heavy hitter items are mapped to the same channel (when
$K$ is sufficiently large).
For each channel %
, users with $H(v_j)=v^*$ encode $v_j$ into
$\vec{x}_j=\mathsf{Enc}(v_j)$ and send the perturbed version of
$\vec{x}_j$;  whereas $\vec{x}_j=\vec{0}$ for users with $H(v_j)\neq
v^*$ and $\mathcal{R}_{\mathsf{bas}}(\vec{0})$ is reported to the
server.

\begin{algorithm}[tp]
	\scriptsize %
	\DontPrintSemicolon
	\KwIn{$m$-bit string $\vec{x}$, privacy budget
		$\varepsilon$}
	Sample $r \gets [m]$ uniformly at random.\;
	\uIf{$\vec{x} \neq \vec{0}$}{%
		$z_r = \begin{cases}
		c\cdot m\cdot x_r & \mbox{
			w.p. $\frac{e^{\varepsilon}}{e^\varepsilon +1}$} \\
		-c\cdot m\cdot x_r & \mbox{ w.p. $\frac{1}{e^\varepsilon + 1}$}
		\end{cases}$,
		where $c=\frac{e^{\varepsilon}+1}{e^{\varepsilon}-1}$.\;
	}
	\Else{%
		Choose $z_r$ uniformly from $\{c\sqrt{m}, -c\sqrt{m}\}$\;
	}
	\Return $\vec{z} = (0, \ldots, 0, z_r, 0, \ldots, 0)$\;
	\caption{$\mathcal{R}_{\mathsf{bas}}(\vec{x}, \varepsilon)$: $\varepsilon$-Basic Randomizer}
	\label{alg:basic_rand}
\end{algorithm}
Given a set of noisy reports $\{\vec{z}_1, \ldots, \vec{z}_n\}$
collected from $n$ users, the server aggregates them to
$\bar{\vec{z}}$ (line~\ref{line:aggr} in Algorithm~\ref{alg:server_side}, in Appendix), rounds it to the nearest valid
encoding $\vec{y}$ (line~\ref{line:rounding_s}-\ref{line:rounding_f} in
Algorithm~\ref{alg:server_side}, in Appendix), and
finally reconstructs the heavy hitter item by decoding it into $\hat{v}
= \mathsf{Dec}(\vec{y})$.
To estimate the frequency of $\hat{v}$, the server collects another
set of noisy reports $\{\vec{w}_1, \ldots, \vec{w}_n\}$ and estimates
the frequency as follows:
\[
  \hat{f}(\hat{v}) = \ip{\frac{1}{n}\sum_{j=1}^n \vec{w}_j, \vec{c}(\hat{v})}
  = \frac{1}{n}\sum_{j=1}^n \vec{w}_j[r_j] \cdot \vec{q}[r_j]\,,
\]
where $\vec{q} = \vec{c}(\hat{v})$.

To filter out possible false positives,
similarly to the previous phase the server collects noisy
reports $\vec{w}_j$ from users and aggregates them in a single bitstring
$\overline{\vec{w}}$.
For each reconstructed value $\hat{v}$ in $\Gamma$, its frequency $f
\left(\hat{v}\right)$ is estimated using a frequency oracle (FO) function.
If the computed estimate $\hat{f}\left(\hat{v}\right)$ is less than a threshold $\eta$, then $\hat{v}$ is removed from $\Gamma$.
After this filtering phase, the server can then return the set of detected heavy
hitters.

The threshold $\eta$ plays a crucial role in the heavy hitter detection:
\begin{equation}
	\label{eq:beta_eta}
	\eta = \dfrac{2T + 1}{\varepsilon}\sqrt{\dfrac{\log(d)\log(1/\beta)}{n}}
\end{equation}
where $\beta$~\cite{Bassily2015} is a parameter related to the confidence the server has on the heavy hitters it has detected. The same parameter $\beta$ also influences the number of protocol rounds, $T$~\cite{Bassily2015}.

We now analyze the properties of the basic randomizer. It is easy to
see that for every encoded item $\vec{x} \in \{\frac{1}{\sqrt{m}},
-\frac{1}{\sqrt{m}}\}^m \cup \{\vec{0}\}$ its noisy report
$\vec{w}=(w_1, \ldots, w_m)$ is an unbiased estimator of
$\vec{x}$. For users with $\vec{x}\neq \vec{0}$ and an integer $r \in [m]$,
\begin{align*}
  \E[w_r]
  &= cm\Bigl(\frac{e^{\varepsilon}}{e^{\varepsilon}+1}
    - \frac{1}{e^\varepsilon + 1}\Bigr)x_r = m\cdot x_r \text{ and } \\
  \E[\vec{w}]
  &= \frac{1}{m}(\E[w_1], \ldots, \E[w_m])^\intercal = \vec{x}\,.
\end{align*}
For users with $\vec{x}=\vec{0}$, $\E[w_r] = 0$ for $\forall r \in [m]$, and
hence $\E[\vec{w}] = \vec{0}=\vec{x}$.
It is also easy to see that for $v \in \mathcal{V}$ the estimated
frequency $\hat{f}(v)$ has the following properties (see
Appendix~\ref{app:basic_rand} for details):
\begin{align*}
  \E[\hat{f}(v)]
  &= f(v)\
\end{align*}
and
\begin{align}
	\label{eq:basic_rand_variance}
	\Var(\hat{f}(v))
	&=
	\frac{1}{n}\biggl\{\biggl(\frac{e^{\varepsilon}+1}{e^{\varepsilon}-1}\biggr)^2
	- f(v)\biggr\}\,.%
\end{align}

\vspace{5pt}
\noindent
{\em Limitations:}
The SH protocol described in~\cite{Bassily2015} was presented in a purely formal way, without addressing limitations that exist in practical systems. For instance, the original SH protocol was formulated in an ``asymptotic'' setting,
in which a large number of reporting clients is assumed. While the protocol works well {\em in expectation}, it presents a number of practical drawbacks, which we discuss in Section~\ref{sec:system_and_attacker_models}.

\subsection{Frequency Oracle Protocol}
\label{sec:OLH}

Wang \textit{et al.}~\cite{OLH} recently proposed the Optimal Local Hashing (OLH) protocol for estimating the frequency of items belonging to a given domain.
It satisfies $\varepsilon$-LDP and is simpler and logically equivalent to the frequency oracle
proposed in~\cite{Bassily2015}.
Instead of transmitting a single bit, that is the result of mapping an item $i$ to a
binary value $\{c\sqrt{m}, -c\sqrt{m}\}$, the $n$ users who participate in
the system simply hash their items into a value in $[g]$, where $g \geq 2$.
The pseudo-code of the OLH randomizer is reported in
Algorithm~\ref{alg:olh_rand} (in Appendix~\ref{apx:background}).
For further details, we refer the reader to~\cite{OLH} and to its publicly
available implementation~\cite{OLHImpl}.
It is worth noting that OLH is limited to frequency estimation, and
that it is not suitable by itself for heavy hitter detection in large domains,
as for the case in which the domain includes all possible valid phone
numbers.

\section{System Details}
\label{sec:system_and_attacker_models}

As mentioned earlier, we envision a collaborative blacklisting system consisting of $n$ distinct smartphones and a centralized server $C$, as shown in Figure~\ref{fig:sys_overview}. $C$ is responsible for receiving data from the participating phones and for computing a blacklist of phone numbers (i.e., caller IDs) likely associated with phone spamming activities. Once computed, the blacklist can be propagated back to the participating phones to enable flagging future unwanted calls as {\em likely spam}.

Each participating smartphone runs an application that collects information about phone calls received from \textit{unknown} phone numbers, where \textit{unknown} here means that the caller's phone number was not registered into the smartphone's contact list. More precisely, let $p_i$ be a participating smartphone and $c_j$ be a caller ID. If $p_i$ receives a call from $c_j$ and $c_j$ is not in $p_i$'s contact list, then $c_j$ is labeled as \textit{unknown} and reported to $C$ by $p_i$. Notice that, in this scenario, only the caller ID $c_j$ is reported, and no information about the content of the call is shared with $C$.

To preserve the privacy of phone calls received by participating users (i.e., the owners of the phones that contribute data to $C$), the caller ID data collection app running on each smartphone implements a local differentially private (LDP) algorithm, whose details are described below in this section.

\subsection{Overview of LDP Protocol}
\label{sec:OverviewLDP}

Following the intuitions and motivations for our approach provided in
Section~\ref{sec:problem_def}, we cast the problem of learning a phone blacklist
as a {\em heavy hitter detection} problem. To this end, we build our solution
upon the SH protocol for heavy hitter detection proposed in~\cite{Bassily2015}
and summarized in Section~\ref{sec:background}. Unfortunately, the original SH
protocol is not directly suitable for our application, because it was formulated
in a theoretical ``asymptotic'' setting in which a large number of reporting
clients is assumed~\cite{Bassily2015}.

First, we implemented a practical client-server formulation of the SH protocol
proposed in~\cite{Bassily2015}. After performing pilot experiments, we found
that applying the protocol {\em as is} in a setting with a limited number of
clients (e.g., in the tens of thousand) and in which we aim to correctly
reconstruct heavy hitters with a relatively low volume (e.g., only a few
hundreds of hits) was not possible without setting an extremely high privacy
budget, thus completely jeopardizing users' privacy.

To make SH usable in practice and adapt it to our phone blacklisting problem, we
therefore designed and implemented a number of modifications that address the
following two fundamental problems:

\begin{itemize}
\item {\em Sparsity of user reports}: In the SH protocol, the larger the items domain $\mathcal{V}$, the more frequent an item must be to be correctly reconstructed by the server $C$ with high probability. Namely, an item must be reported by a larger and larger population, as the cardinality of $\mathcal{V}$ increases, thus potentially impeding the reconstruction of spam phone numbers involved in campaigns that reach only a portion of the contributing users.

\item {\em High variance in the sum of item reports hinders noise cancellation}: The sum $\overline{\vec{z}}$ in Algorithm~\ref{alg:server_side} (line 6) is affected by the high variance of the distribution of the sum of each bit. Ideally, the noisy random bits sent by users who do not hold value $v$ (i.e.,that transmit a randomized version of $\vec{x} = \vec{0}$ in Algorithm~\ref{alg:client_side}, lines 7-8) should cancel out during summation. While this holds {\em in expectation}, in practice (with a finite number of participants) it is highly unlikely to have the very same number of clients who transmit $\frac{1}{\sqrt{m}}$ as clients who transmit $-\frac{1}{\sqrt{m}}$, potentially causing the reconstruction of the wrong bit value at the server side.
\end{itemize}

We address the first issue by introducing a data bucketization mechanism. Specifically, we take advantage of characteristics of the phone blacklisting problem to (1) reduce the dimensionality of the items domain, and (2) partition the problem domain to increase the relative frequency of heavy hitters. To achieve (1), we divide phone numbers into area code prefix (e.g., the first three digits in a US telephone number) and phone number suffix (e.g., the remaining seven digits, for a US phone number).

In telephony, area codes are typically not considered to be sensitive. For instance, the FTC dataset protects the privacy of complaining users by publishing only their area code~\cite{FTC_DNC} (see also Sections~\ref{sec:problem_def} and~\ref{sec:telephony_dataset}).
As outlined in Section~\ref{sec:problem_def}, we aim to protect the privacy of the phone number suffix.
Hence, given the large number of phone numbers that share the same prefix,
clients can transmit the area code {\em as is}, and apply the (modified) SH
protocol only to the phone number suffix, thus reducing the number of bits
needed to represent each phone number. To achieve (2), we assign a separate
communication channel between clients and server to each area code, and run an
instance of the (modified) SH protocol independently for each area code. This
has the effect of clustering phone numbers based on their prefix. Because in
some cases phone spam campaigns are conducted using specific area codes (e.g., a
Washington DC area code for IRS spam campaigns, or an 800 prefix for tech
support scams, etc.), this {\em bucketization} of phone numbers has the effect
of amplifying the relative frequency of spam-related caller IDs in some of the
area code buckets (or clusters), thus making it easier to detect heavy hitters.
Figure~\ref{fig:after_buckets_amplification} shows a example of how in practice
bucketization helps to amplify the relative frequency of heavy hitters.

\begin{figure}[b!]
	\centering
	\includegraphics[width=\linewidth]{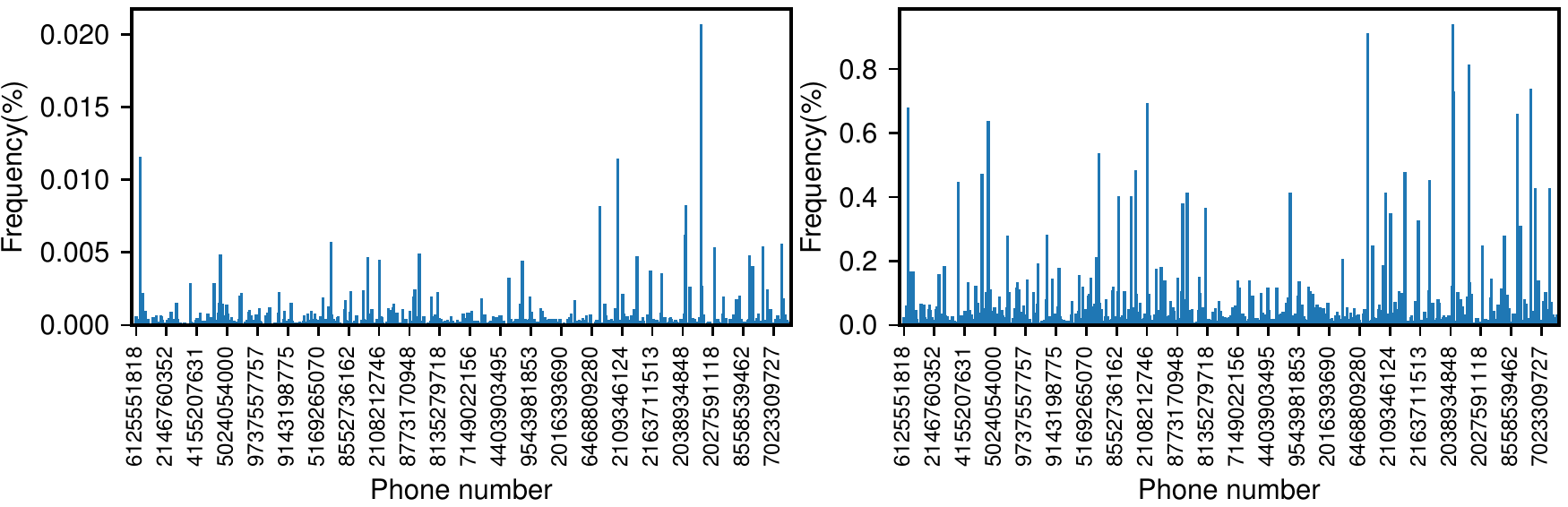}
	\vspace{-5pt}
	\caption{Phone number frequency before (left) and after (right) bucketization, computed on one day of complaints from the FTC dataset.}
	\label{fig:after_buckets_amplification}
\end{figure}

In summary, we (i) group phone numbers by area code, (ii) split area code and phone suffix, (iii) select the client-server communication channel based on the area code, and (iv) let each client transmit the area code in clear (i.e., no LDP) and run the SH protocol over phone number suffixes within transmitted area codes.
\begin{algorithm}[t]
	\scriptsize %
	\DontPrintSemicolon
	\caption{Modified SH-Client($v$, $\mathcal{H}$, $T$, $K$, \textcolor{black}{$\varepsilon_\mathsf{HH}$, $\varepsilon_\mathsf{OLH}$)}}
	\label{alg:client_side_variant}
	\KwIn{the value to be sent $v$, a fixed list of hash functions $\mathcal{H}$,
		\# of repetitions $T$, \# of channels $K$, privacy parameter\textcolor{black}{s $\varepsilon_\mathsf{HH}$ and
			$\varepsilon_\mathsf{OLH}$}}
	\BlankLine
	\tcc{sending noisy reports for heavy hitter detection}
	$\rho \gets prefix\left(v\right)$ \;
	$\sigma \gets v \setminus \rho$ \;
	\SetKwFor{For}{\textcolor{mygray}{for}}{\textcolor{mygray}{do}}{end for}
	\SetKwFor{ForEach}{\textcolor{mygray}{foreach}}{\textcolor{mygray}{do}}{end foreach}
	\For{\textcolor{mygray}{$t=1$ \KwTo $T$}}{%
		\textcolor{mygray}{$H \gets \mathcal{H}[t]$}\;
		\ForEach{\textcolor{mygray}{channel $k \in [K]$}}{
			\uIf{$H(\sigma) = k$}{%
				$\vec{x} = \mathsf{Enc}(\sigma)$\;
			}%
			\textcolor{mygray}{\uElse{%
					$\vec{x} = \vec{0}$\;
			}}%
			$\vec{z}^{(t,k,\rho)} \gets \mathcal{R}_{\mathsf{ext}}\left(\vec{x}, \dfrac{\varepsilon_\mathsf{HH}}{2T}\right)$\;
			Send $\vec{z}^{(t,k,\rho)}$ to the server on channel $k$\;
		}
	}
	\tcc{sending noisy report for heavy hitter frequency estimation}
	$w^{(\rho)} \gets \mathcal{R}_{\mathsf{OLH}}\left(v, \varepsilon_\mathsf{OLH}\right)$\;
	Send $w^{(\rho)}$ to the server\;
\end{algorithm}
\begin{algorithm}[tp]
	\scriptsize %
	\DontPrintSemicolon
	\caption{Modified SH-Server($T$, $K$, $\mathrm{P}$, $\mathsf{OLH}$)}
	\label{alg:server_side_variant}
	\KwIn{\# of repetition $T$, \# of channels $K$, set of prefixes
		$\mathrm{P}$, a frequency oracle
		$\mathsf{OLH}$, threshold $\eta$}
	\KwOut{list of heavy hitters $\Gamma$}
	\tcc{detecting heavy hitters}
	\textcolor{mygray}{$\Gamma \gets \emptyset$} \;
	\SetKwFor{For}{\textcolor{mygray}{for}}{\textcolor{mygray}{do}}{end for}
	\For{\textcolor{mygray}{$t=1$ \KwTo $T$}}{%
		\ForEach{prefix $\rho \in [\mathrm{P}]$}{%
			\SetKwFor{ForEach}{\textcolor{mygray}{foreach}}{\textcolor{mygray}{do}}{end foreach}
			\ForEach{\textcolor{mygray}{channel $k \in [K]$}}{%
				\SetKwFor{ForEach}{\textcolor{black}{foreach}}{\textcolor{black}{do}}{end foreach}
				\ForEach{user $j \in [n_\rho]$}{%
					$\vec{z}_j \gets \vec{z}^{(t,k,\rho)}$ value received from user $j$ on channel $k$, having prefix $\rho$;
				}
				$\overline{\vec{z}} = \frac{1}{n_\rho}\sum_{j=1}^{n_\rho} \vec{z}_j$\;
				\For{$i=1$ \textcolor{mygray}{\KwTo} $m$}{%
					\textcolor{mygray}{$\vec{y}[i] \gets
						\begin{cases}
						\frac{1}{\sqrt{m}} & \mbox{if $\overline{\vec{z}}[i] \geq 0$}
						\\
						-\frac{1}{\sqrt{m}} & \mbox{ otherwise.}
						\end{cases}$}\;
				}%
				$\hat{\sigma} \gets \mathsf{Dec}(\vec{y})$\;
				$\hat{v} \gets$ append $\sigma$ to $\rho$ \;
				\textcolor{mygray}{\lIf{$\hat{v} \notin \Gamma$}{add $\hat{v}$ to $\Gamma$}}
			}%
		}%
	}
	\tcc{filtering out false positives}
	\ForEach{prefix $\rho \in [\mathrm{P}]$}{%
		\ForEach{user $j \in [n_\rho]$}{%
			$\vec{\overline{w}}[j] \gets w^{(\rho)}$ value received from user $j$ having prefix $\rho$\;
		}
		\SetKwFor{ForEach}{\textcolor{mygray}{foreach}}{\textcolor{mygray}{do}}{end foreach}
		\ForEach{\textcolor{mygray}{$\hat{v} \in \Gamma$}}{%
			$\hat{f}(\hat{v}) \gets $ estimate the frequency of $\hat{v}$ using $\mathsf{OLH}(\overline{\vec{w}})$\;
			\textcolor{mygray}{\lIf{$\hat{f}(\hat{v}) < \eta$}{remove $\hat{v}$ from $\Gamma$}}
		}%
	}
	\textcolor{mygray}{\Return $\{(v, \hat{f}(v))~:~v \in \Gamma\}$} \;
\end{algorithm}

To address the high variance in the sum of item reports, we introduce
a new {\em extended randomizer} to replace the original randomizer
proposed in~\cite{Bassily2015} and reported in
Algorithm~\ref{alg:basic_rand}. The main idea is to use a three-value
randomizer. For instance, when $\vec{x} = \vec{0}$ must be sent,
instead of choosing a random bit value between $\{\frac{1}{\sqrt{m}},
-\frac{1}{\sqrt{m}}\}$, the client app will choose between three
values: $\{\frac{1}{\sqrt{m}}, 0, -\frac{1}{\sqrt{m}}\}$, with different
probabilities. Our externded randomizer is defined in Algorithm~\ref{alg:ext_rand}.
In Appendix~\ref{app:ext_rand_analysis} we formally show how
this extended randomizer helps reducing the variance, thus increasing the accuracy
with which privately-reported phone numbers are reconstructed on the
server side.

\begin{algorithm}[bp]
	\scriptsize %
	\DontPrintSemicolon
	\KwIn{$m$-bit string $\vec{x}$, privacy budget
		$\varepsilon$}
	Sample $r \gets [m]$ uniformly at random.\;
	\uIf{$\vec{x} \neq \vec{0}$}{%
		$z_r = \begin{cases}
		c\cdot m\cdot x_r
		& \mbox{ w.p. $p$}\\
		-c\cdot m\cdot x_r
		& \mbox{ w.p. $q$} \\
		0
		& \mbox{ w.p. $1-p-q$}
		\end{cases}$, where $c>0$.
	}
	\Else{%
		$z_r = \begin{cases}
		c\sqrt{m}
		& \mbox{ w.p. $\theta$}\\
		-c\sqrt{m}
		& \mbox{ w.p. $\theta$} \\
		0
		& \mbox{ w.p. $1-2\theta$}
		\end{cases}$
	}
	\Return $\vec{z} = (0, \ldots, 0, z_r, 0, \ldots, 0)$\;
	\caption{$\mathcal{R}_{\mathsf{ext}}(\vec{x}, \varepsilon)$: $\varepsilon$-extended Randomizer}
	\label{alg:ext_rand}
\end{algorithm}

Notice also that while in the following we present our LDP protocol under the assumption that each user has a single item to share with the server (e.g., one unknown phone number report per day), in real scenarios some users may either have multiple items to share or nothing to share at all (e.g., no phone calls received in a given day).
In this case, the protocol can be easily extended as proposed in~\cite{Qin2016}, by sampling a single telephone number from the set of unknown calls that the app has collected. Conversely, if a user has nothing to share, the app can generate a dummy (but legitimate) phone number to be sent to the server.

It is also important to notice that organizing phone number reports in buckets allows the server to count the number of users that will participate to a protocol run, per each bucket.
In Appendix~\ref{sec:ac_buckets}, we discuss how the server can use this number to estimate the probability that at least one heavy hitter in a specific bucket will be successfully reconstructed.
If such probability is low, the server can avoid executing the protocol for those buckets and inform the clients of this decision, thus preventing those clients from wasting their privacy budget. In practice, once the server receives the area codes from each client, it could send a message back to the clients letting them know if they should send (using the LDP protocol) the remaining portion of the caller IDs they observed (i.e., the remaining seven digits) or not.

The new LDP protocol resulting from our improvements over the original SH protocol is shown in Algorithms~\ref{alg:client_side_variant} and~\ref{alg:server_side_variant}, where new pseudo-code is highlighted in black, and code that remains the same as in the original SH protocol is shaded in gray.
Unlike the original version, we explicitly allocate two different privacy budgets, $\varepsilon_\mathsf{HH}$ and $\varepsilon_\mathsf{OLH}$, assigned respectively to heavy hitter detection and frequency estimation.
It is worth mentioning that $\varepsilon_\mathsf{HH}$ is the total privacy budget spent by each user to send noisy reports to the server during the $T$ protocol rounds (lines $3-11$).
In this new formulation $\varepsilon~= \left(\varepsilon_\mathsf{HH} +
\varepsilon_\mathsf{OLH}\right)$ and the protocol is
$\left(\varepsilon_\mathsf{HH} + \varepsilon_\mathsf{OLH}\right)$-differentially
private, as proved in Appendix~\ref{app:ext_rand_analysis}.

\section{LDP Protocol Evaluation}
\label{sec:evaluation}

In this section, we present an evaluation of our LDP protocol. It is important to notice that we focus  primarily on estimating the accuracy of our system with respect to {\em reconstructing and detecting heavy hitters}. We will discuss the utility of the phone blacklist that may be learned from the detected heavy hitters separately, in Section~\ref{sec:utility}.

\subsection{Dataset}
\label{sec:telephony_dataset}

Ideally, to evaluate our protocol we would need to collect a dataset of phone call records from thousands of users. Even though we assume only calls from {\em unknown} numbers (i.e., numbers not stored in the users' respective contact lists) would be of interest for detecting potential spam phone numbers, collecting such a dataset is very difficult, due exactly to the same privacy issue we aim to solve in this paper. As a proxy for that dataset, we make use of real-world phone data extracted from user complaints to the FTC.
In essence, the FTC allows users in the US to report unwanted phone calls. Reported complaints that are made available to the public typically include the time of the complaint, the full caller ID, the user's phone number area code, and a label indicating the type of phone spam activity.
Notice that the FTC aim to protect the privacy of the users who report a complaint; that's why they only publish users' phone number prefixes. On the other hand, the caller IDs are reported in clear because the complaining users explicitly label them as {\em unwanted caller}, essentially consenting to their public release. On the other hand, our system does not require users to explicitly label unwanted phone calls; all {\em unknown} caller IDs are reported, and privacy is preserved via LDP mechanisms (please refer to Section~\ref{sec:intro}, where we motivate the advantages of this approach).

In this paper, we treat each complaint as if a user participating in our collaborative blacklist learning system had received a phone call from the complained-about number at the time recorded in the complaint.
We were able to obtain a large set of user complaints collected by the FTC between Feb. 17th 2016 and Mar. 17th, 2016, for a total of $29$ days, which consists of $471,460$ complaints.
For the sake of this evaluation, we consider only valid 10-digit caller IDs, which constitute about 95\% of the entire dataset.
The distribution of complaints is characterized by two properties: (a) the
volume of complaints follows a weekly pattern, with fewer complaints submitted
on weekends; and (2) the distribution of complaints per caller ID has a long
tail, whereby most phone numbers receive only one complaint but there also exist
many phone numbers that receive hundreds of complaints (see
Figures~\ref{fig:ftc_dataset_stats} and~\ref{fig:complaints_distribution} in
appendix for details).

As the FTC dataset does not contain an identifier for the reporting users, without loss of generality we assume that, within a day, each complaint is reported by a different user. Based on this, we determine the pool of users that would participate in our collaborative blacklisting as follows: we compute the maximum number of reports seen in one day, throughout the entire month of FTC data, which is equal to $23,188$; we then set the number of participants to that number. In days in which fewer than $23,188$ users sent complaints to the FTC, we assume that the remaining users did not have any calls to report (i.e., they did not receive any calls from {\em unknown} numbers on those days). However, to preserve differential privacy guarantees, {\em all} users must send a report every time the protocol runs (e.g., daily, in our evaluation). Therefore, if a user has no calls to report, the app on her smartphone will generate a random (but valid) 10-digit number, and report that number to the server using the LDP protocol, exactly as if she received a call from that number.

\begin{figure*}[t!]
	\centering
  \includegraphics[scale=0.7]{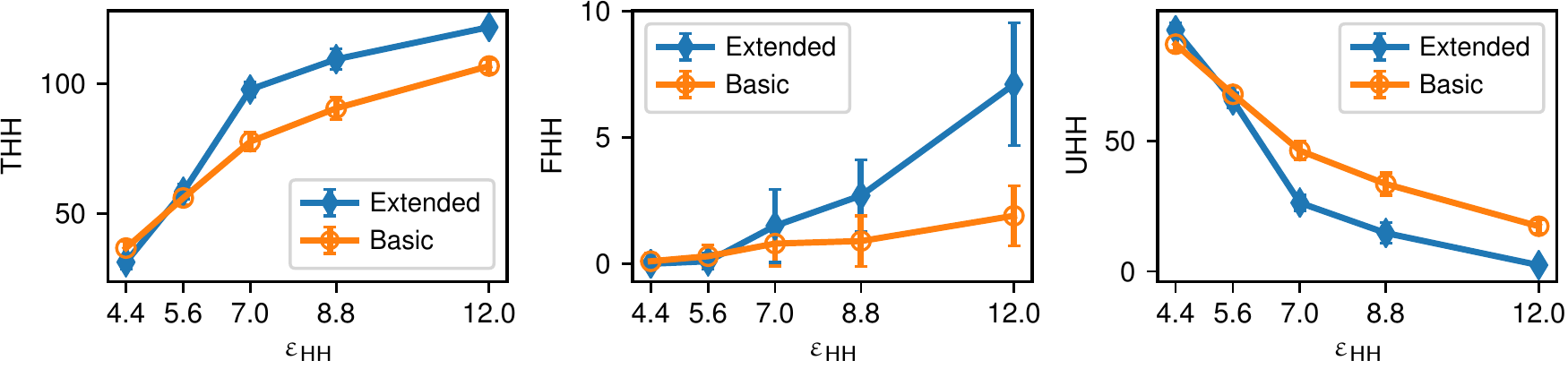}
	\vspace{-10pt}
	\caption{True, false and undetected heavy hitters (respectively, THHs, FHHs and UHHs). Parameters: $T=2$, $\varepsilon_\mathsf{OLH} = 3$, and $\tau = 143$.}
	\label{fig:tpr_fp}
\end{figure*}
\begin{figure*}[t!]
	\centering
    \includegraphics[scale=0.7]{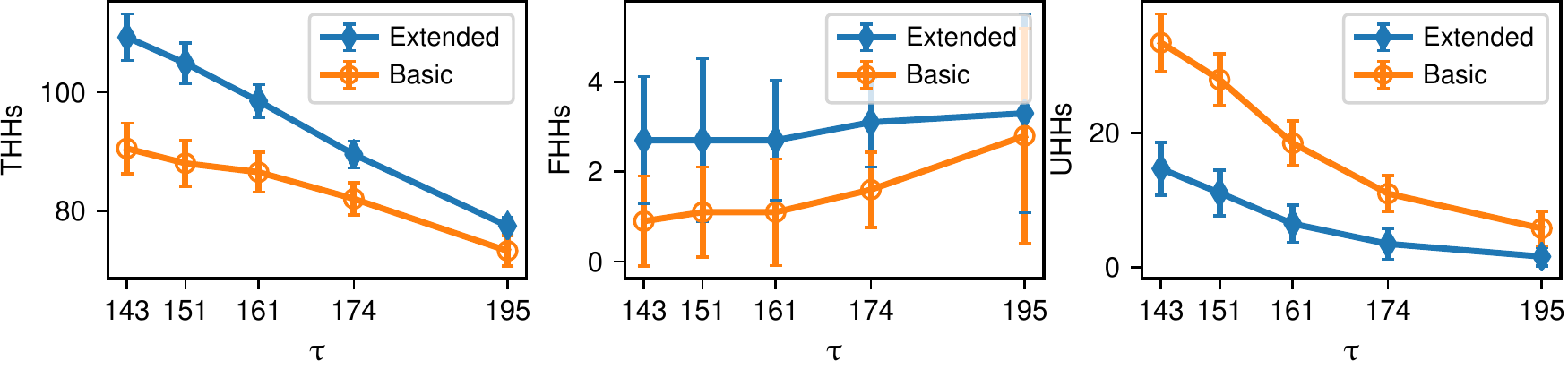}
	\vspace{-10pt}
	\caption{True, false and undetected heavy hitters (respectively, THHs, FHHs and UHHs). Parameters: $T=2$, $\varepsilon_\mathsf{HH} = 8.8$, $\varepsilon_\mathsf{OLH} = 3$.}
	\label{fig:tpr_fp_tau}
\end{figure*}

\subsection{LDP Protocol Configuration}
\label{sec:BS_variation}

In all our experiments we use area code bucketization, with the same parameter settings for all buckets. We also compare results obtained using the basic randomizer proposed in~\cite{Bassily2015} to results obtained using our extended randomizer (Algorithm~\ref{alg:ext_rand}), using the same parameter settings for each, to allow for an ``apples to apples'' comparison.

We set two different privacy budgets to run the heavy hitter detection
and the frequency estimation protocol. Respectively, we experiment
with budget $\varepsilon_\mathsf{HH}~\in~\{12, 8.8, 7, 5.6, 4.4\}$ for
heavy hitter detection, and $\varepsilon_\mathsf{OLH} = 3$ for
frequency estimation. We chose the values $\varepsilon_\mathsf{HH}$
for both the basic and extended randomizer so that the randomizer
sends the correct phone number with probability approximately equal to
$0.95$, $0.90$, $0.85$, $0.80$ and $0.75$. 
Moreover, we do not allocate a $\varepsilon_\mathsf{HH}$ value lower than
$4.4$ given that, as later discussed in Section~\ref{sec:discussion},
the utility of the learned blacklist is already $0$ when
$\varepsilon_\mathsf{HH} = 4.4$ (see
Figure~\ref{fig:sh_variant_cbr}). 
It is important to also notice that guaranteeing differential privacy under continual observation~\cite{DPcontinual} in the more difficult LDP setting is still an open problem in differential privacy research. A complete solution to such a challenging open problem is therefore left to future work. Nonetheless, one possible mitigation may be to design the data collection app so that a user does not report the same number more than once (see Section~\ref{sec:discussion} for further discussion).

To implement the binary encoding of phone numbers (see Algorithm~\ref{alg:client_side_variant} line $7$), we use a Reed Muller error-correcting code, $RM(3,5)$; this is a $\left[32,26,4\right]_2$ code with relative distance equal to $1/8$ and error correcting capability equal to $1$ bit~\cite{guruswami2004}. Notice that $k=26$ bits is sufficient to encode 7-digit phone numbers, which requires a minimum of $24$ bits.
At server-side, we run the heavy hitter detection phase only for those buckets containing more than a minimum number, $\tau$, of complaints, as motivated in Sections~\ref{sec:OverviewLDP} and~\ref{sec:ac_buckets}.
We experiment with $5$ different values of $\tau$ in the set $\{143, 151, 161, 174, 195\}$. These values correspond, respectively, to a $75\%$, $80\%$, $85\%$, $90\%$, $95\%$ probability of correctly reconstructing, at least, a phone number per bucket (see also Equation~\ref{eq:non_empty_cells_distribution}).

\subsection{Measuring Heavy Hitter Detections}
\label{sec:measuring_accuracy}

We now define how we measure accuracy for the LDP protocol. Notice that in this section we consider accuracy strictly for the {\em heavy hitter detection} task accomplished by the LDP protocol. This is related to but different from the utility of the blacklist that can be learned over the phone numbers reconstructed by the server-side LDP protocol (see Section~\ref{sec:utility} for results on blacklist utility).

Let $v$ be a phone number, and $c(v)$ be the number of users who reported a call from $v$. We say that $c(v)$ is the {\em ground truth frequency} of $v$. Moreover, let $\hat{f}(v)$ be the number of reports about $v$ estimated by the server after running the LDP protocol, and $\eta$ be the detection threshold for heavy hitter detection defined in Equation~\ref{eq:beta_eta} (see also Algorithms~\ref{alg:server_side} and~\ref{alg:server_side_variant}).
Also, as discussed in Section~\ref{sec:BS_variation}, let $\tau$ be the minimum number of complaints necessary for the server to decide whether to run the heavy hitter detection phase for a bucket.

In theory, we could simply use $\eta$ as heavy hitter frequency threshold, to measure true and false detections. However, given Equation~\ref{eq:beta_eta} and substituting practical values of $\varepsilon$ and $\beta$, $\eta$ tends to be much smaller than $\tau$. For instance, considering $\varepsilon = 15$, $\beta = 0.751$, $T=2$, and $d = 10^7$, and assuming $n=1000$ users reporting caller IDs to a given area code bucket, we obtain $\eta\approxeq0.023$. Thus, the heavy hitter detection threshold (in terms of number of reports per caller ID) would be $\eta \cdot n = 23$. In other words, a phone number would be considered a heavy hitter if it is reported more than 22 times. Yet, as we discussed in Section~\ref{sec:ac_buckets}, the minimum number of reports needed for the server to correctly reconstruct a phone number with high probability is much higher than 23 (e.g., at least 84 reports are needed to have a 50\% chance of correct reconstruction). We therefore use $\tau$ as the heavy hitter detection threshold, rather than relying on $\eta$.
Specifically, we define the following quantities.

\begin{itemize}

\item \textbf{True Heavy Hitters} (THHs). We have a {\em true heavy hitter} detection for $v$ if both $c(v) > \textcolor{black}{\tau}$ and $\hat{f}(v) > \textcolor{black}{\tau}$.

\item \textbf{False Heavy Hitters} (FHHs). We have a false heavy hitter detection for $v$ if $c(v) \leq \textcolor{black}{\tau}$ whereas $\hat{f}(v) > \textcolor{black}{\tau}$.

\item \textbf{Undetected Heavy Hitters} (UHHs). We have an undetected heavy hitter if $c(v) > \textcolor{black}{\tau}$ whereas $\hat{f}(v) \leq \textcolor{black}{\tau}$.

\end{itemize}

It is worth noting that FHHs are typically due to a phone number $v$ whose true frequency $c(v)$ is just below $\tau$, and for which the noise introduced by the LDP protocol causes the server to (by chance) estimate its frequency above the heavy hitter detection threshold. %
On the other hand, UHHs represent heavy hitters that the protocol fails to detect, due to the random noise added by the clients. %
Given the above definitions, and their analogy with true and false positives in detection systems, we measure the F1-score of the heavy hitter detection protocol as:

\begin{itemize}

\item \textbf{Recall}: $ R = THHs / (THHs + UHHs) $

\item \textbf{Precision}: $ P = THHs / (THHs + FHHs) $

\item \textbf{F1-score}: $F_1 = 2*(P*R) / (P+R)$

\end{itemize}

\subsection{LDP Heavy Hitter Detection Accuracy}
\label{sec:detection_accuracy}

Figure~\ref{fig:tpr_fp} shows the number of THH detected for different privacy budgets $\varepsilon_\mathsf{HH}$, when $T = 2$, and $\varepsilon_\mathsf{OLH} = 3$ (error bars represent one standard deviation). The figure compares the accuracy that can be obtained by using the basic randomizer (as in~\cite{Bassily2015}) and our extended randomizer (Algorithm~\ref{alg:ext_rand}).
The maximum privacy budget spent daily by each client running the LDP protocol can be computed by summing privacy budgets $\varepsilon_\mathsf{HH}$ and $\varepsilon_\mathsf{OLH}$ (e.g., $\varepsilon = 15$, when $\varepsilon_\mathsf{HH} = 12$ and $\varepsilon_\mathsf{OLH} = 3$). It is worth noting that an user may have no phone number to report: in that case, the privacy budget spent by the client would be $0$.
Each experimental evaluation with a given $\varepsilon_\mathsf{HH}$ and $\varepsilon_\mathsf{OLH}$ was repeated $10$ times, and the results averaged.

\begin{figure}[t]
  \centering
  \includegraphics[scale=0.60]{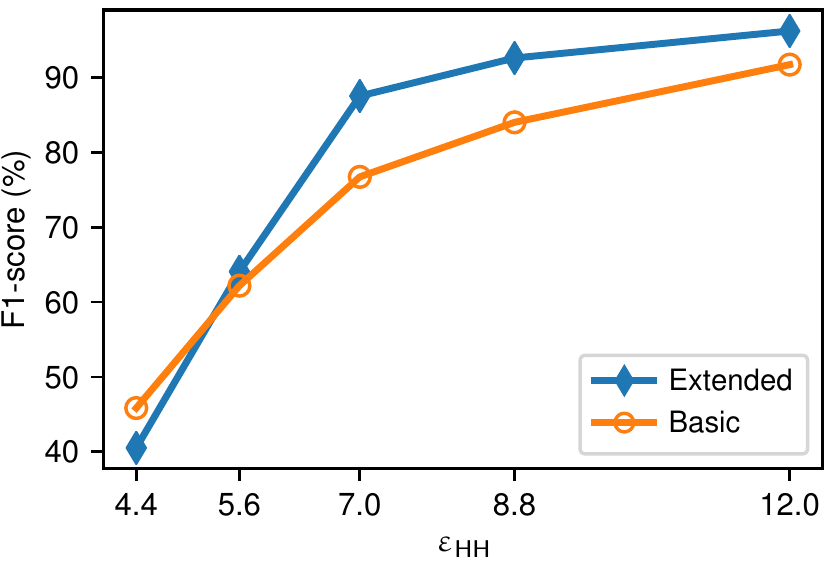}
  \vspace{-10pt}
  \caption{F1-score with parameters $T=2$, $\varepsilon_\mathsf{OLH} = 3$, and $\tau~=~143$.}
	\label{fig:sh_f1_T=2}
\end{figure}

Figure~\ref{fig:tpr_fp} also reports the number of FHHs and UHHs obtained for different values of the  privacy budget. As can be seen, the LDP protocol detects less than 8 FHHs, on average.
As mentioned in Section~\ref{sec:measuring_accuracy}, such false heavy hitters are phone numbers whose frequency is just below $\tau$ and whose LDP-estimated frequency happens to slightly exceed the heavy hitter detection threshold due to the randomization of user contributions.
In addition, the higher the $\varepsilon_\mathsf{HH}$ allocated for detecting heavy hitters, the higher the probability of correctly reconstructing phone numbers whose frequency is just below $\tau$ and, hence, generating FHHs.
Figure~\ref{fig:tpr_fp_tau} shows how THHs, FHHs, and UHHs vary with $\tau$, using the same parameters of the previous experimental evaluations, but fixing $\varepsilon_\mathsf{HH}$ to $8.8$.
As can be observed, increasing $\tau$ decreases the total number of detectable heavy hitters (i.e., the sum of THHs and UHHs), as expected, since fewer and fewer reported caller IDs will have a true frequency $c(v) > \tau$.

\begin{figure}[t]
	\centering
	\includegraphics[scale=0.60]{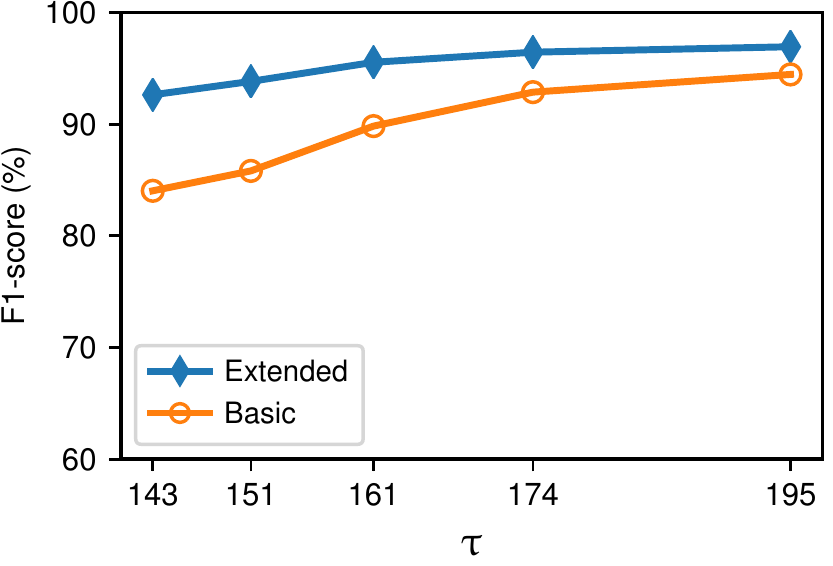}
	\vspace{-10pt}
	\caption{F1-score with parameters: $T=2$, $\varepsilon_\mathsf{HH} = 8.8$, $\varepsilon_\mathsf{OLH} = 3$, and $\tau = 143$.}
	\label{fig:sh_f1_T=2_tau}
\end{figure}

It is also important to notice that, as shown in Figure~\ref{fig:sh_f1_T=2}, overall our LDP protocol with the extended randomizer performs better than using the basic randomizer proposed in~\cite{Bassily2015}, when $\varepsilon_\mathsf{HH} \in \{12, 8.8, 7\}$, which yield an F1-score above $85\%$.
The scores have been computed by using THHs, FHHs, and UHHs depicted in Figure~\ref{fig:tpr_fp}, averaged across $10$ runs.
For lower values of  $\varepsilon_\mathsf{HH}$, the F1-score decreases significantly, and the extended randomizer tends to perform slightly worse than the basic randomizer. This is because the probability of injecting noise in the reports (at the clients side) increases considerably.
This aspect, in combination with the fact that, by definition, the extended randomizer has a lower probability of sending the correct report to the server, compared the basic randomizer, determines a slight reduction in performance for low values of $\varepsilon_\mathsf{HH}$.

Finally, Figure~\ref{fig:sh_f1_T=2_tau} reports the F$1$-score as $\tau$ changes. Higher values of $\tau$  allow us to obtain higher scores, because the number of UHHs decreases (see Figure~\ref{fig:tpr_fp_tau}). The basic and extended randomizers follow similar trends, though the extended randomizer performs better than the basic one, independently from the choice of $\tau$.

\section{Blacklist Utility}
\label{sec:utility}

In Section~\ref{sec:evaluation} we evaluated the ability of our LDP protocol to accurately detect heavy hitter caller IDs. We now look at how a blacklist learned over heavy hitters detected using our protocol would fare compared to when no privacy is preserved, whereby caller IDs are collected from users' phones and sent directly to the server (no noise added). To compare these scenarios, we leverage the {\em call blocking rate} (CBR) metric proposed in~\cite{Pandit2018}.

In a way similar to~\cite{Pandit2018}, we define a blacklist $\mathbb{B}$ as a set of caller IDs that have been reported by users more than $\theta$ times. Specifically, as in~\cite{Pandit2018}, we use a sliding window mechanism, whereby a blacklist $\mathbb{B}$ is updated daily by cumulatively adding {\em daily heavy hitter} caller IDs observed over the past time window (one week, in our experiments). Blacklisted caller IDs older than the sliding window are forgotten, and removed from $\mathbb{B}$. As an example, making use again of the FTC dataset (see Section~\ref{sec:telephony_dataset}), the blacklist that each user deploys on February
$24$th contains all the heavy hitters detected each day during the week
going from Feb. $17$th to Feb. $23$rd. The CBR is then computed by measuring how many calls are flagged by $\mathbb{B}$ on the day of deployment.

To enable a comparison between the private and non-private versions of blacklist learning, we set the same fixed heavy hitter detection threshold $\theta$ for both. In other words, in the case when no privacy is offered, caller IDs that are reported by more than $\theta$ users in a day are considered as potential spammers. Similarly, when our LDP protocol is used to learn the blacklist, we fix the heavy hitter detection threshold $\tau = \theta$. The other protocol parameters in this experiment are set to $T=2$ and $\varepsilon_\mathsf{OLH} = 3$, while varying $\varepsilon_\mathsf{HH}$.

As a baseline, we compute (over the FTC dataset) the median call blocking rate CBR$^*$ that can be achieved throughout a month of FTC reports, without applying any privacy-preserving mechanism and for different values of $\theta$ (we compute the median because it is less sensitive to outliers, compared to the average). Then, we compare CBR$^*$ to the CBR obtained by the blacklist learned using our LDP protocol, by computing the median of the fraction of calls that our blacklist would block, compared to CBR$^*$. The results are reported in Figure~\ref{fig:sh_variant_cbr}. As can be seen, as the overall privacy budget $\varepsilon$ increases, the CBR %
approaches the baseline CBR$^*$, which is indicated by the 100\% mark.
It can be noticed that the difference with the baseline increases as $\theta$ reduces. This is because it is more unlikely that the server will correctly reconstruct caller IDs that have a lower number of reports (see Section~\ref{sec:ac_buckets}). Therefore, as $\theta$ decreases, heavy hitters with low frequency (close to $\theta$) can still be detected in the scenario without privacy, but become more difficult to detect for our LDP protocol.
\begin{figure}[t]
	\centering
	\includegraphics[scale=0.56]{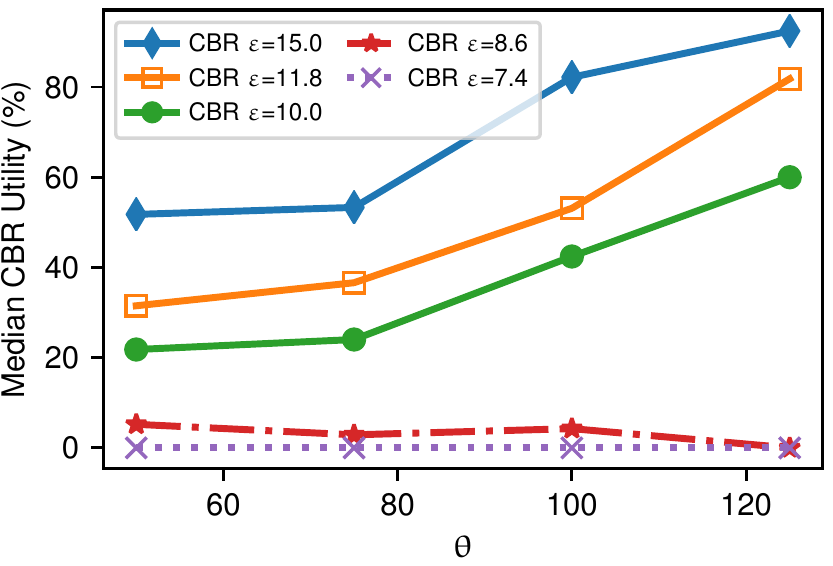}
	\vspace{-5pt}
	\caption{CBR: percentage of calls blocked compared to the baseline.}
	\label{fig:sh_variant_cbr}
\end{figure}

In practice, whenever a user receives a call from an {\em unknown} caller ID that is in the blacklist, the app will inform the user that the number is suspicious, and potentially involved in spamming. The user may ultimately decide to pick up the call, but use more caution when interacting with the other party.

\section{Discussion}
\label{sec:discussion}

In Section~\ref{sec:detection_accuracy}, we have
reported several results related to the accuracy of the proposed LDP protocol using
different privacy budgets and confidence parameters.
Depending on how much budget the server provides to system users, the SH protocol
parameters can be tuned to control privacy/utility trade-offs.
As mentioned in Section~\ref{sec:BS_variation}, Apple uses up to $\varepsilon = 16$~\cite{Apple2016} as privacy budget
for gathering statistics. For instance, the Safari browser allows for two user
contributions per day, with $\varepsilon=8$ each. On the other hand, in this
paper we experimented with a maximum privacy budget of $\varepsilon=15$ with one
user contribution per day. 
While the privacy budget may seem somewhat high compared to {\em non-local}
differential privacy applications, it is worth noting that this is due to the
inherent complexity of LDP. For instance, it has been shown that
$\varepsilon$-LDP distribution estimators require $k/\varepsilon^2$ times larger
datasets than a comparable non-private algorithm~\cite{Duchi2013,
kairouz2016discrete}, where $k$ is the size of the input alphabet (i.e., $k$ is
the number of possible phone number combinations, in our case). As $k$ can be
very large, higher values of $\varepsilon$ allow us to achieve an acceptable
utility even with relatively small values of the sample set size $n$ (i.e., the
number of noisy reports received by the server).
Furthermore, we also showed that even for lower values of epsilon (e.g., $\varepsilon=11.8$), blacklist utility can still be reasonable (e.g., around 80\% of the CBR$^*$ obtained in the scenario with no privacy, as shown in Section~\ref{sec:utility}). A privacy budget that can provide more privacy while keeping a good performance trade-off is  $\varepsilon = 10$ (with $T=2$, $\varepsilon_\mathsf{HH} = 7$, $\varepsilon_\mathsf{OLH} = 3$, and $\tau = 143$): our experimental evaluation shows an F1-score higher than $75\%$ with the detection of more than $97$ potential spam phone numbers per day, on average.

A limitation of our system, which is common to practical deployments of LDP such as in the case of Apple and other vendors, is that guaranteeing differential privacy under continual observation~\cite{DPcontinual} in an LDP setting is still an open research problem in differential privacy. As a possible mitigation, the data collection app running on the user's phone can keep history of the reported numbers and avoid reporting the same calling number more than once within a given time window (e.g., one month). This would make it much more difficult for the server to identify a phone number that may have called a specific user with high frequency (e.g, once a day), since it will be reported only once by that user. At the same time, if the same number is reported only once but by many users, it can still be detected as heavy hitter and added to the blacklist.
It is also possible that a legitimate phone number may be reported by many users, such as in the case of school alert numbers or other types of emergency phone numbers that may contact a large number of users at once, since these numbers may not be recorded in every user's contact list. Such phone numbers may potentially be detected as heavy hitters, and thus considered by the server for blacklisting. However, the server could check the validity of a number, before propagating it to the blacklist. For instance, the server could use automated {\em reverse phone number lookup} services (e.g., \url{whitepages.com}) to filter out possible false positives related to emergency numbers.
Our work is based on the heavy hitter LDP protocol proposed in~\cite{Bassily2015}, which, to the best of our knowledge, was one of very few state-of-the-art LDP protocols for heavy hitter detection at the time when we started the research presented in this paper. Alongside~\cite{Bassily2015}, RAPPOR~\cite{Erlingsson:2014,Fanti:2016} is another protocol that could be adapted to fit our problem.
However, it has been shown that RAPPOR performs less well than a more recent protocol named TreeHist~\cite{Bassily2017}, and that in turn TreeHist itself has a higher worst-case error, compared to the original SH protocol proposed in~\cite{Bassily2015}. Similarly, it has been shown in~\cite{OLH} that for frequency estimation the OLH protocol (which we summarized in Section~\ref{sec:OLH} and used in our system) performs better than RAPPOR.

Recently, a few new LDP protocols for heavy hitter detection have been also proposed~\cite{Bassily2017,Wang2017PEM,Qin2016}. However, \cite{Bassily2015} remains a state-of-the-art protocol that has inspired more recent works. Furthermore, in this paper we focus on studying how to make LDP heavy hitter detection practical to address an important and previously unsolved security problem: {\em privacy-preserving collaborative phone blacklisting}.
We believe that the application-specific trade-offs between privacy and utility we presented in this paper would still be relevant even if \cite{Bassily2015} was replaced by a different LDP heavy hitter detection protocol.

In Section~\ref{sec:evaluation}, we performed experiments with a fixed value of parameter $T=2$. In the original formulation of the SH protocol~\cite{Bassily2015}, $T$ is directly related to the parameter $\beta$ we mentioned in Section~\ref{sec:background}. While it would be possible in theory to use higher parameter values, increasing $T$ (by varying $\beta$) would result in a higher number of protocol rounds, and would thus consume a much larger privacy budget $\varepsilon$ for each user. Conversely, increasing $T$ while keeping $\varepsilon$ fixed would cause a significant degradation of heavy hitter detection accuracy, and in turn of the blacklist utility. Therefore, for the sake of brevity, we did not report experimental results obtained with larger values of $T$.
	\section{Related Work}
\label{sec:rel_work}

Besides RAPPOR~\cite{Erlingsson:2014,Fanti:2016}, which we briefly discussed in Section~\ref{sec:discussion}, there exist other works related to LDP heavy hitter detection; we briefly discuss them below. However, it should be noted that our work is different from the ones discussed here. Our main contributions are in adapting a state-of-the-art protocol proposed in~\cite{Bassily2015} to make it practical, and in using the adapted protocol to build a collaborative phone blacklisting system with provable privacy guarantees.

In~\cite{Qin2016}, the SH
protocol proposed in~\cite{Bassily2015} is extended to handle set-valued data, where each
user holds a set of items $\vec{v} = \{v_1,\ldots, v_t\} \subseteq
\mathcal{V}$. One difficulty in the set-valued data setting is that
the length of the itemset each user has is different. To address this
challenge, Qin et al~\cite{Qin2016} proposed a protocol, called
LDPMiner, for finding heavy hitters from set-valued data. The main idea of
LDPMiner is to pad each user's itemset with dummy items to ensure that
it has the fixed length $\ell$. Each user randomly samples one item
from $\vec{v}$ and reports the item using the SH protocol. The
estimated frequency of items in LDPMiner is multiplied by $\ell$ to
account for the random sampling procedure.

Bassily et al.~\cite{Bassily2017} and Wang et al.~\cite{Wang2017PEM}
independently proposed a similar protocol that
iteratively identifies heavy hitters using a prefix tree. In their
protocol, users are randomly split into $g$ disjoint groups. At
iteration $i$, the server receives noisy reports from the users
in the $i$\th group, . Each user in the $i$\th group reports the
randomized version of the first $l_i$ bits of the encoded item (i.e.,
a prefix of length $l_i$), where $l_1<l_2 \cdots < l_g$. After
aggregating the user reports from the $i$\th group, the server
identifies frequent prefixes $C_i$ of length $l_i$ and builds the
candidate heavy hitter items of length $l_{i+1}$ by concatenating
$C_i$ with strings in $\{0, 1\}^{l_{i+1}-l_{i}}$.
Recently, Wang et al.~\cite{Wang2018FIM} provided a thorough analysis
on the ``pad-and-sampling-based frequency oracle (PSFO)'' and proposed
an LDP solution to the frequent itemset mining problem.
Their protocol adaptively chooses between two algorithms based on the size
of the domain $|\mathcal{V}|$.

	\section{Conclusion}
We proposed a novel collaborative detection system that learns a list of spam-related phone numbers from call records contributed by participating users. Our system makes use of local differential privacy to provide clear privacy guarantees. We evaluated the system on real-world user-reported call records collected by the FTC, and showed that it is possible to learn a phone blacklist in a privacy preserving way using a reasonable overall privacy budget, while at the same time maintaining the utility of the learned blacklist.

	\bibliographystyle{abbrv}
	\bibliography{references}
	
	\appendix

\section{Client-Server SH Algorithms}
\label{apx:background}

Algorithms~\ref{alg:client_side} and~\ref{alg:server_side} show the
client-server formulation of the SH protocol discussed in Section~\ref{sec:BS_proto}.

In order to run the client protocol, each client first needs to know the number of communication channels $K$ that has to be established with the
server for sending private reports.
Hence, before starting the SH protocol,
the server communicates the correct number of channels $K$ to the clients.
Notice that the server is the only one who can compute $K$, since $K$ depends
on the number of users contributing to the system
at any given time.
For $T$ times, in each channel $k$ and round $t$ in $[T]$,
the user sends to the server a randomized report $\vec{z}^{(t,k)}$, which represents the (encoded) value of $v$ she holds or a special value
\textbf{0} indicating that the user does not hold a value to be reported.

The choice of sending the randomized report associated with $\mathsf{Enc}(v)$ (or with \textbf{0}) depends
on whether the channel identifier $k$ matches the value returned by the hash function $H$ applied on $v$.
$H$ belongs to a pairwise independent hash function family $\mathcal{H}$, publicly available and
accessible to all the clients as part of the client-side protocol configuration.

In each round of the protocol, a different hash function is employed to minimize the probability of
collisions among different heavy hitters.
Notice that, except for a single channel in which the client sends the private report
obtained from $\mathcal{R}_{\mathsf{bas}}$ for a value $v$, in all the other
channels the client sends randomized reports for the special value \textbf{0}
(see Algorithm~\ref{alg:client_side}).

\begin{algorithm}[b!]
	\footnotesize %
	\DontPrintSemicolon
	\KwIn{value $v$, a hash function $H$, OLH $g$ parameter, privacy budget~$\varepsilon$}
	$x \gets H(v)\ \%\ g $\;
	Sample $y \gets [g] \setminus \{x\}$ uniformly at random.\;
	$w = \begin{cases}
	x & \mbox{
		w.p. $\frac{e^{\varepsilon}}{e^\varepsilon+g-1}$} \\
	y & \mbox{ w.p. $\frac{g-1}{e^\varepsilon+g-1}$}
	\end{cases}$\;
	
	\Return $w$\;
	\caption{$\mathcal{R}_{\mathsf{OLH}}(v, \varepsilon)$: $\varepsilon$-OLH Randomizer}
	\label{alg:olh_rand}
\end{algorithm}

On the server side, the server receives in each channel $k$ the
private reports $\vec{z}^{(t,k)}$ sent by users for each specific run
$t$. In each round and for each channel, the server aggregates
the randomized reports to reconstruct the codeword $\vec{y}$ whose hash of the
original value $v$ corresponds to channel $k$.
Hence, the decoded value $\hat{v}$, if correctly reconstructed, should
represent the private information sent by (a non-negligible number of)
users in the $k$-th channel in a specific run of the SH protocol.
The set of reconstructed values is stored in the set of potential
heavy hitters $\Gamma$. Due to noisy reports, some values in $\Gamma$
may not be heavy hitters.

To filter out possible false positives,
similarly to the previous phase the server collects noisy
reports $\vec{w}_j$ from users and aggregates them in a single bitstring
$\overline{\vec{w}}$.
For each reconstructed value $\hat{v}$ in $\Gamma$, its frequency $f
\left(\hat{v}\right)$ is estimated using a frequency oracle (FO) function.
If the computed estimate $\hat{f}\left(\hat{v}\right)$ is less than a threshold $\eta$, then $\hat{v}$ is removed from $\Gamma$.
After this filtering phase, the server can then return the set of detected heavy hitters.

The threshold $\eta$ plays a crucial role in the heavy hitter detection:
\begin{equation}
	\label{eq:beta_eta}
	\eta = \dfrac{2T + 1}{\varepsilon}\sqrt{\dfrac{\log(d)\log(1/\beta)}{n}}
\end{equation}
where $\beta$~\cite{Bassily2015} is a parameter related to the confidence the server has on the heavy hitters it has detected. The same parameter $\beta$ also influences the number of protocol rounds, $T$~\cite{Bassily2015}.
The server-side protocol pseudo-code is represented in
Algorithm~\ref{alg:server_side}, whereas Algorithm~\ref{alg:olh_rand} refers to
the discussion in Section~\ref{sec:OLH}.

\section{Analysis of the Basic Randomizer}
\label{app:basic_rand}

\begin{algorithm}[tp]
	\footnotesize %
	\DontPrintSemicolon
	\caption{SH-Client($v$, $\mathcal{H}$, $T$, $K$, $\varepsilon$)}
	\label{alg:client_side}
	\KwIn{the $m$-bit string representation \vec{v} of the value $v$ to be sent, a fixed list of hash functions $\mathcal{H}$,
		\# of repetitions $T$, \# of channels $K$, privacy parameter $\varepsilon$ }
	\BlankLine
	\tcc{sending noisy reports for heavy hitter detection}
	
	\For{$t=1$ \KwTo $T$}{%
		$H \gets \mathcal{H}[t]$\;
		\ForEach{channel $k \in [K]$}{
			\uIf{$H(v) = k$}{%
				$\vec{x} = \mathsf{Enc}(v)$\;
			}%
			\uElse{%
				$\vec{x} = \vec{0}$\;
			}%
			$\vec{z}^{(t,k)} \gets \mathcal{R}_{\mathsf{bas}}\left(\vec{x}, \dfrac{\varepsilon}{2T + 1}\right)$\;
			Send $\vec{z}^{(t,k)}$ to the server on channel $k$\;
		}
		\tcc{sending noisy report for heavy hitter frequency estimation}
		$\vec{w} \gets \mathcal{R}_{\mathsf{bas}}\left(\vec{v},
		\dfrac{\varepsilon}{2T + 1}\right)$\; \label{line:bas_freqreport}
		Send $\vec{w}$ to the server\;
	}
\end{algorithm}

We first show that the frequency estimate $\hat{f}(v)$ obtained using
the basic randomizer is unbiased:
\begin{footnotesize}
\begin{align*}
  \E[\hat{f}(v)]
  &= \E\bigl[\frac{1}{n}\sum_{j=1}^n
    \vec{w}_j^\intercal \vec{x}_v\bigr] \\
  &=\frac{1}{n}\biggl\{\sum_{j:v_j=v} \E[\vec{w}_j^\intercal\vec{x}_v]
    + \sum_{j:v_j\neq v} \E[\vec{w}_j^\intercal\vec{x}_v]\biggr\}\\
  &= \frac{1}{n}
    \sum_{j:v_j=v}\vec{x}_j^\intercal \vec{x}_v \\
  &= \frac{1}{n}\sum_{j:v_j =v} \norm{\vec{x}_j}^2
    = \frac{\sum_{j:v_j=v} 1}{n}=f(v)\,,
\end{align*}
\end{footnotesize}
where $\vec{x}_v=\vec{c}(v)$ denotes the encoding of item $v$.

We next calculate the variance of the estimate given by the basic
randomizer. Let $\vec{r}=(r_1, \ldots, r_j, \ldots, r_n)$ be a vector
of random bits chosen by user $j$, where $r_j \in [m]$.
By the law of total variance, for an item $v\in \mathcal{V}$, the
variance of estimate $\hat{f}(v)$ is
\begin{footnotesize}
\begin{align*}
  \Var(\hat{f}(v))
  &=\E[\Var(\hat{f}(v)~|~\vec{r})] + \Var(\E[\hat{f}(v)~|~\vec{r}]) \\
  &=\frac{1}{n}\{(c^2 - 1)f(v) + (1-f(v)) c^2\} \\
  &= \frac{c^2 - f(v)}{n}\,,
\end{align*}
\end{footnotesize}
where we have
\begin{footnotesize}
\begin{align*}
  &\Var(\hat{f}(v)~|~\vec{r})\\
  &=\Var\bigl(\frac{1}{n}\sum_{j=1}^n w[r_j]\cdot \vec{x}_v[r_j]~|~\vec{r}\bigr) \\
  &=\frac{1}{n^2} \Var\bigl(\sum_{j=1}^n w[r_j]~|~r_j\bigr) \vec{x}_v[r_j]^2 \\
  &=\frac{\vec{x}_v[r_j]^2}{n^2} \biggl\{
    \sum_{j: v_j=v}\Var(w[r_j]~|~r_j) + \sum_{j:v_j\neq v}
    \Var(w[r_j]~|~r_j) \biggr\}  \\
  &= \frac{\vec{x}_v[r_j]^2}{n^2}\biggl\{nf(v)(c^2m^2x[r_j]^2 -
    m^2x[r_j]^2) \\
  &\qquad + n(1-f(v))(c^2m -  0^2)\biggr\}
\end{align*}
\end{footnotesize}
and
\begin{footnotesize}
\begin{align*}
  \E[\hat{f}(v)~|~\vec{r}]
  &=\E\biggl[\frac{1}{n}\sum_{j=1}^n w[r_j] \cdot \vec{x}_v[r_j]~|~\vec{r}\biggr] \\
  &=\frac{1}{n}\biggl\{
    \sum_{j:v_j=v}\E\bigl[w[r_j] \cdot \vec{x}_v[r_j]~|~r_j\bigr] +
    \sum_{j:v_j\neq v} 0\biggr\} \\
  &=\frac{1}{n}\sum_{j:v_j=v} m\cdot x[r_j]^2 \,.
\end{align*}
\end{footnotesize}

\begin{algorithm}[tp]
	\footnotesize %
	\DontPrintSemicolon
	\caption{SH-Server($T$, $K$, $\mathsf{FO}$)}
	\label{alg:server_side}
	\KwIn{\# of repetition $T$, \# of channels $K$, a frequency oracle
		$\mathsf{FO}$, a threshold $\eta$}
	\KwOut{list of heavy hitters $\Gamma$}
	\tcc{detecting heavy hitters}
	$\Gamma \gets \emptyset$ \;
	\For{$t=1$ \KwTo $T$}{%
		\ForEach{channel $k \in [K]$}{%
			\ForEach{user $j \in [n]$}{%
				$\vec{z}_j \gets \vec{z}^{(t,k)}$ value received from user $j$ on channel $k$;
			}
			$\overline{\vec{z}} = \frac{1}{n}\sum_{j=1}^n \vec{z}_j$\;\label{line:aggr}
			\For{$i=1$ \KwTo $m$}{\label{line:rounding_s}%
				$\vec{y}[i] \gets
				\begin{cases}
				\frac{1}{\sqrt{m}} & \mbox{if $\overline{\vec{z}}[i] \geq 0$}
				\\
				-\frac{1}{\sqrt{m}} & \mbox{ otherwise.}
				\end{cases}$\;\label{line:rounding_f}
			}%
			$\hat{v} \gets \mathsf{Dec}(\vec{y})$\;
			\lIf{$\hat{v} \notin \Gamma$}{add $\hat{v}$ to $\Gamma$}
		}%
	}
	\tcc{filtering out false positives}
	\ForEach{user $j \in [n]$}{%
		$\vec{w}_j \gets \vec{w}$ value received from user $j$\;
	}
	$\overline{\vec{w}} = \frac{1}{n}\sum_{j=1}^n \vec{w}_j$\; \label{line:agg_report}
	\ForEach{$\hat{v} \in \Gamma$}{%
		$\hat{f}(\hat{v}) \gets $ estimate the frequency of $\hat{v}$ using $\mathsf{FO}(\overline{\vec{w}})$\;
		\lIf{$\hat{f}(\hat{v}) < \eta$}{remove $\hat{v}$ from $\Gamma$}
	}%
	\Return $\{(v, \hat{f}(v))~:~v \in \Gamma\}$\;
\end{algorithm}

\section{Analysis of Area Code Bucketization}
\label{sec:ac_buckets}

Let us first analyze how the probability that the server $C$ correctly reconstructs a reported phone number depends on the size of the phone numbers space and the number of reports. In this simplified analysis, we will assume no noise is added to the data transmitted from the clients to the server. In other words, we will follow the fundamental steps of the SH protocol in Algorithm~\ref{alg:client_side}, but pretend that the randomizer (line 8) always returns the {\em true value} of one randomly selected bit.

Let us now consider a domain $\mathcal{V}$, in which each value can be represented using $l$ bits (i.e., $|\mathcal{V}| = 2^l$). Also, let us consider a value $v \in \mathcal{V}$ transmitted by $n$ clients. For the sake of this simplified analysis, on the server side we can view $v$ as a sequence of $l$ different {\em bins} that are initially empty, and the bits sent by the clients as {\em balls}. According to the SH protocol for heavy hitter detection, each client transmits only one bit, and therefore the server receives $n$ balls. To correctly reconstruct the value $v$, at least one ball must fill each bin. As reported in~\cite{BBL12}, the number of non-empty bins resulting from randomly inserting $n$ balls into $l$ bins has the following probability distribution:
\begin{equation} \label{eq:non_empty_cells_distribution}
U_{l,n}(b) = \frac{ \stirling{n}{b} \binom{l}{b} b! }{l^n}, \quad \forall b \in \{1, \ldots, l\}
\end{equation}
where $\stirling{n}{b}$ is a Stirling number of the second kind, which expresses the number of ways to partition a set of $n$ elements into $b$ non-empty subsets. The numerator in the equation expresses the number of ways in which $n$ balls fall exactly in $b$ bins out of $l$ available ones. Therefore, for $b = l$,
$ U_{l,n} = \frac{ \stirling{n}{l} l! }{l^n} $
gives us the probability that all bins will be filled.

Intuitively, the larger $l$, the larger $n$ must be to fill the bins. For instance, in this simplified analysis, the 34-bit representation of a 10-digit phone number $p$ would need to be reported by at least 170 users, for it to have about an 80\% probability of being reconstructed at the server side. In reality, the additional noise and the error-correction encoding in the SH protocols further complicate the relationship between $l$ and $n$. However, it is clear that reducing $l$ also reduces the number of reports above which heavy hitters can be detected with high probability. This motivates our choice of {\em bucketizing} phone numbers by grouping them based on area codes, and by running a separate instance of the SH protocol per bucket, as only seven digits need to be reported by the SH protocol for each phone number in a bucket. Following the above analysis, 111 reports are sufficient to reconstruct 24-bit values (needed to represent 7-digit numbers) with 80\% probability, which equates to about a 34.7\% reduction in the number of reports to be received by the server.

As outlined in Section~\ref{sec:OverviewLDP}, Equation~\ref{eq:non_empty_cells_distribution} can also be used by the server for deciding if the clients that have a report to be sent within a given bucket (i.e., if they need to report a caller ID within a given prefix) should actually send the report (using LDP) or not. Considering 24 bits per phone number, as above, and assuming all clients in the same bucket intend to report the same 7-digit phone number, all buckets receiving less than $84$ reports can be easily ignored, because the server will have less than $50\%$ probability of correctly reconstructing a heavy hitter in those buckets. This probability is even lower in practice, since each bucket will likely receive reports about different phone numbers. Instructing clients that intend to send a report to ``low density'' buckets to stop doing so will prevent running the LDP protocol in vain. Thus, those clients can avoid wasting their privacy budget for those specific LDP protocol runs.

Another benefit of grouping phone numbers by area code is that some spam
campaigns tend to use numbers with specific area codes.
Figure~\ref{fig:bucket_distro} visually shows this tendency.

Figure~\ref{fig:after_buckets_amplification} shows a more comprehensive view of how the relative frequency of phone numbers in the FTC data is amplified when bucketization is used. Specifically, each vertical line represents the frequency of caller IDs appearing in the FTC complaints dataset. The figure on the left shows the occurrence frequency of phone numbers relative to all complaints received in one day, whereas the figure on the right shows how their relative frequency changes after bucketization (notice the different y-axis scales for the two graphs).
The take away from this analysis is that bucketization results in the amplification of the relative frequency of some heavy hitter caller IDs and, hence, in the variance reduction of frequency estimates (see Equation~\ref{eq:basic_rand_variance}), thus increasing the likelihood that heavy hitters will be correctly reconstructed and detected by the server.

\section{Analysis of Extended Randomizer}
\label{app:ext_rand_analysis}

While the frequency estimate $\hat{f}(v)$ of an item $v \in
\mathcal{V}$ computed from noisy reports generated using the basic
randomizer (in line~\ref{line:bas_freqreport} of
Algorithm~\ref{alg:client_side}) is unbiased,
its variance is often quite large in practice, and this
could lead to low accuracy in heavy-hitter detection.
Inspired by the antithetic variates technique in Monte Carlo methods~\cite{Kleijnen2013},
we extend the basic randomizer
and introduce a new randomizer $\mathcal{R}_{\mathsf{ext}}$ which
yields lower variance. The extended randomizer is described in
Algorithm~\ref{alg:ext_rand}.

The main difference between the randomizers is in the number of
different values each user can report. Notice that $z_r \in
\{c\sqrt{m}, 0, -c\sqrt{m}\}$ in the extended randomizer, while $z_r
\in \{c\sqrt{m}, -c\sqrt{m}\}$ in the basic randomizer. The idea
behind this modification is that the sum of contributions from
users who don't have item $v$ to the estimate $\hat{f}(v)$ is non-zero
in practice, due to the variance, while in expectation they should cancel
out.

The following lemma shows that the extended randomizer provides an unbiased
estimate of (encoded) item $\vec{x}$.
\begin{restatable}[]{lemma}{lemextrand} \label{lem:extrand}
  Let $p = \frac{e^{\epsilon}}{e^{\epsilon}+2}$,
  $q=\theta=\frac{1}{e^{\epsilon}+2}$, and
  $c=\frac{e^{\epsilon}+2}{e^{\epsilon}-1}$. The extended randomizer
  $\mathcal{R}_{\mathsf{ext}}$ has the following properties:
  \begin{enumerate}[label=(\roman*)]
  \item For every $\vec{x} \in \{-1/\sqrt{m}, 1/\sqrt{m}\}\cup
    \{\vec{0}\}$, $\E[\mathcal{R}_{\mathsf{ext}}(\vec{x})] = \vec{x}$.
  \item $\mathcal{R}_{\mathsf{ext}}$ satisfies $\epsilon$-LDP for
    every $r \in [m]$.
  \end{enumerate}
\end{restatable}
The proof of the above lemma is provided in Appendix~\ref{app:ext_rand}.

Given a set of noisy reports $\vec{z}_1, \ldots, \vec{z}_n$ generated
by the extended randomizer, the randomizer yields an unbiased
estimate of frequency with smaller variance than the basic
randomizer. The following lemma formalizes this discussion, whose
proof appears in Appendix~\ref{app:ext_rand}.
\begin{restatable}[]{lemma}{lemfreqest} \label{lem:freqest}
Let $v^*\in\mathcal{V}$ be an item and $\{\vec{w}_i\}_{i=1}^n$ be the
noisy reports. The frequency estimate $\hat{f}(v^*) = \frac{1}{n}\sum_{j=1}^n
\vec{w}_j^\intercal\vec{c}(v^*)$ has the following properties:
\begin{enumerate}[label=(\roman*)]
\item $\E[\hat{f}(v^*)] = f(v^*)$ and
\item $\Var(\hat{f}(v^*)) = \frac{1}{n}\bigl\{
  f(v^*)\cdot (c^2(p+q)-1) + (1 - f(v^*)) \cdot 2c^2\theta\bigr\}$,
\end{enumerate}
where $f(v^*)$ is the true frequency of $v^*$.
\end{restatable}

Two important remarks are in order. First, the extended randomizer
$\mathcal{R}_{\mathsf{ext}}$ reduces to the basic randomizer
$\mathcal{R}_{\mathsf{bas}}$ if we set
$c=\frac{e^{\varepsilon}+1}{e^{\varepsilon}-1}$,
$p=\frac{e^{\varepsilon}}{e^{\varepsilon}+1}$,
$q=\frac{1}{e^{\varepsilon}+1}$, and $\theta=\frac{1}{2}$.
Second, the above shows that the
variance of frequency estimate of an item $v^*\in \mathcal{V}$ can be
written as a linear combination of two terms:
$c^2(p+q)$ and $2c^2\theta$. While we wish to find optimal parameter
values for $c, p, q$, and $\theta$ that minimize the variance, this is
not possible because $f(v^*)$ is unknown. Instead, we minimize the
maximum of those two terms under $\varepsilon$-LDP constraints:

\begin{footnotesize}
\[
\begin{aligned}
& \underset{c, p, q, \theta}{\text{minimize}} & &
\max\,\{c^2(p+q),\, 2c^2\theta\}\\
& \text {subject to} & & c(p-q)=1 \\
& & & p - e^{\epsilon}\theta \leq 0\,,\;
-p + e^{-\epsilon}\theta  \leq 0 \\
& & & q - e^{\epsilon}\theta \leq 0\,,\;
-q + e^{-\epsilon}\theta \leq 0 \\
& & & p - e^{\epsilon} q \leq 0 \,,\;
-p - e^{-\epsilon}q \leq 0 \\
& & & 1 - p - q - e^{\epsilon}(1 - 2\theta) \leq 0\\
& & & -1 + p + q + e^{-\epsilon}(1 - 2\theta) \leq 0 \\
& & & 0 \leq p + q \leq 1\,,\;
0 \leq \theta \leq \frac{1}{2}\,.
\end{aligned}
\]
\end{footnotesize}
Solving the above optimization problem gives the following solution:
\begin{footnotesize}
\begin{equation} \label{eq:ext_rand_param}
p=\frac{e^{\varepsilon}}{e^{\varepsilon}+2}\,,\quad
q=\theta=\frac{1}{e^{\varepsilon}+2}\,,\quad
c=\frac{e^{\varepsilon}+2}{e^{\varepsilon}-1}\,.
\end{equation}
\end{footnotesize}

\begin{restatable}[]{proposition}{proplowvar} \label{prop:lower_variance}
  The frequency estimate $\hat{f}(v)$ of an item $v$ given by
  $\mathcal{R}_{\mathsf{ext}}$ has lower variance than that given by
  $\mathcal{R}_{\mathsf{bas}}$ if
  \[
    \varepsilon \geq \ln\frac{a + \sqrt{9a^2 -20a + 12}}{1-a}\,,
  \]
  where $a = f(v)$, i.e., the true frequency of $v$.
\end{restatable}
The proof of the above proposition is simple and given in Appendix~\ref{app:ext_rand}.

\begin{restatable}[]{theorem}{thmprivacy} \label{thm:privacy}
  Algorithm~\ref{alg:server_side_variant} satisfies
  \textcolor{black}{$\left(\varepsilon_\mathsf{HH} + \varepsilon_\mathsf{OLH}\right)$}-differential privacy.
\end{restatable}
The proof of Theorem~\ref{thm:privacy} follows from~\cite[Theorem
3.4]{Bassily2015} and is included in the Appendix~\ref{app:ext_rand} for completeness.

\subsection{Proofs for Extended Randomizer}
\label{app:ext_rand}

\lemextrand*
\begin{proof}
  Consider an item $v^*\in \mathcal{V}$ on a channel $k\in [K]$ and a
  hash function $H:\mathcal{V} \to [K]$. For users $j$ with $H(v_j) =
  k$, we have
  \begin{footnotesize}
  \begin{align*}
    \E[\mathcal{R}_{\mathsf{ext}}(\vec{x}_j)]
    &= \E[\vec{z}_j] =\E\bigl[\E[\vec{z}_j~|~r_j]\bigr] \\
    &= \frac{1}{m}(\E[\vec{z}_j[1]], \ldots, \E[\vec{z}_j[m]])^\intercal
    \\
    &=\frac{1}{m}\bigl(cm(p-q)\vec{x}_j[1], \ldots,
      cm(p-q)\vec{x}_j[m]\bigr)^\intercal \\
    &=c(p-q
    +)\vec{x}_j\,.
  \end{align*}
  \end{footnotesize}
  Since $c(p-q)=1$, we have $\E[\mathcal{R}_{\mathsf{ext}}(\vec{x}_j)] = \vec{x}_j$.
  For those users with $H(v_j)\neq k$, their encoded item
  $\vec{x}_j=\mathsf{Enc}(v_j)=\vec{0}$, and we have
  \begin{footnotesize}
  \begin{align*}
    \E[\vec{z}_j]
    &= \frac{1}{m}(c\sqrt{m}\theta - c\sqrt{m}\theta, \ldots,
      c\sqrt{m}\theta-c\sqrt{m}\theta)=\vec{0}=\vec{x}_j\,.
  \end{align*}
  \end{footnotesize}
  This completes the proof of the unbiasedness of
  $\mathcal{R}_{\mathsf{ext}}$.

  Next, we prove $\epsilon$-LDP of the extended randomizer. Let
  $v_1$ and $v_2$ be two arbitrary items in $\mathcal{V}$ and
  $\vec{x}_1$ and $\vec{x}_2$ be their encodings in $\{-1/\sqrt{m},
  1/\sqrt{m}\}^m \cup \{\vec{0}\}$, respectively.
  For any $z_r \in \{cmx_r, 0, -cmx_r\}$, we have
  \[
    \frac{\Pr[z_r~|~\vec{x}_1, r]}{\Pr[z_r~|~\vec{x}_2, r]}
    \leq \max\left\{\frac{p}{\theta},
      \frac{1-2\theta}{1 - p - q}\right\}  = e^{\epsilon}\,.
  \]
  Similarly,
  \[
    \frac{\Pr[z_r~|~\vec{x}_1, r]}{\Pr[z_r~|~\vec{x}_2, r]} \geq
    \min\left\{\frac{1-p-q}{\theta}, \frac{\theta}{p}\right\} = e^{-\epsilon}\,.
  \]
\end{proof}

\lemfreqest*
\begin{proof}
  Let $\vec{x}^* = \vec{c}(v^*)$.
  We first prove the unbiasedness property. Since $\vec{w}_j$ is an
  unbiased estimate of $\vec{x}_j$ (i.e., $\E[\vec{w}_j] =
  \vec{x}_j$), it is easy to see that $\hat{f}(v^*)$ is also
  unbiased.
  \begin{footnotesize}
  \begin{align*}
    \E[\hat{f}(v^*)]
    &=\E\biggl[\frac{1}{n}\sum_{j=1}^n
      \vec{w}_j^{\intercal}\vec{c}(v^*)\biggr] \\
    &=\frac{1}{n}\Bigl\{\sum_{j:v_j=v^*}\E[\vec{w}_j^\intercal\vec{x}_j]
      + \sum_{j:v_j\neq v^*}\E[\vec{w}_j^\intercal\vec{x}^*]\Bigr\}\\
    &=\frac{1}{n}\sum_{j:v_j=v^*}\norm{\vec{x}_j}^2 =
      \frac{\sum_{j:v_j=v^*} 1}{n} = f(v^*)\,.
  \end{align*}
  \end{footnotesize}
  To compute the variance $\Var(\hat{f}(v^*))$, we condition on random
  bits chosen by users.  Let $\vec{r}=(r_1, \ldots, r_j, \ldots, r_n)$
  be a vector, where $r_j \in [m]$ represents the random bit chosen by
  user $j$.
  By the law of total variance,
  \begin{footnotesize}
  \begin{align}
    \Var(\hat{f}(v^*))
    &=\E\bigl[\Var(\hat{f}(v^*)~|~\vec{r})\bigr]
      +\Var\bigl(\E[\hat{f}(v^*)~|~\vec{r}]\bigr) \nonumber \\
    &=\frac{1}{n^2}\E\bigl[
      \Var\bigl(\sum_{j=1}^n \vec{w}[r_j]\cdot \vec{x}^*[r_j]
      ~|~r_j\bigr)\bigr] \nonumber \\
    &\qquad + \frac{1}{n^2}\Var\biggl(
      \E\Bigl[\sum_{j=1}^n \vec{w}[r_j]\cdot
      \vec{x}^*[r_j]~|~r_j\Bigr]\biggr) \nonumber \\
    &=\frac{1}{n^2}(\E[A] + \Var(B))\,.
      \label{eq:var_aplusb}
  \end{align}
  \end{footnotesize}
  The first term is
  \begin{footnotesize}
  \begin{align*}
    A
    &=\sum_{j=1}^n \Var(\vec{w}[r_j]~|~r_j)\cdot \vec{x}^*[r_j]^2 \\
    &=\sum_{j:v_j=v^*} \Var(\vec{w}[r_j]) \cdot \vec{x}^*[r_j]^2 +
      \sum_{j:v_j\neq v^*} \Var(\vec{w}[r_j])\cdot \vec{x}^*[r_j]^2 \\
    &=\sum_{j:v_j=v^*} \bigl(c^2m^2\vec{x}[r_j]^2(p + q) -
      m^2\vec{x}[r_j]^2\bigr) \cdot \vec{x}^*[r_j]^2\\
    &\qquad +
      \sum_{j:v_j\neq v^*} 2c^2m\theta\cdot \vec{x}^*[r_j]^2\,,
  \end{align*}
  \end{footnotesize}
  and
  \begin{footnotesize}
  \begin{align}
    \E[A]
    &=nf(v^*)\cdot \frac{1}{m}\Bigl(c^2m^2(p+q)\sum_{i=1}^m \vec{x}_j[i]^4 -
      m^2\sum_{i=1}^m \vec{x}_j[i]^4\Bigr) \nonumber \\
    & + n(1-f(v^*)\cdot \frac{1}{m} \cdot 2c^2m\theta \sum_{i=1}^m
      \vec{x}^*[i]^2 \nonumber \\
    &=nf(v^*)\{c^2(p+q) - 1\} + n(1-f(v^*))\cdot 2c^2\theta\,.
      \label{eq:var_parta}
  \end{align}
  \end{footnotesize}
  The second term is
  \begin{footnotesize}
  \begin{align*}
    B
    &=\sum_{j=1}^n \E\bigl[\vec{w}[r_j] \cdot \vec{c}(v^*)[r_j]~|~r_j\bigr] \\
    &=\sum_{j:v_j=v^*} \E\bigl[\vec{w}[r_j]\bigr]\cdot \vec{x}^*[r_j] +
      \sum_{j:v_j\neq v^*} \E\bigl[\vec{w}[r_j]\bigr] \cdot\vec{x}^*[r_j] \\
    &=\sum_{j:v_j=v^*} cm\vec{x}[r_j]^2(p - q) = nf(v^*) \cdot
      cm\vec{x}[r_j]^2(p - q)\,,
  \end{align*}
  \end{footnotesize}
  and
  \begin{footnotesize}
  \begin{align}
    \Var(B)
    &=n^2f(v^*)^2c^2m^2(p-q)^2\Var(\vec{x}[r_j]^4) = 0\,.
      \label{eq:var_partb}
  \end{align}
  \end{footnotesize}
  Plugging~\eqref{eq:var_parta} and~\eqref{eq:var_partb}
  into~\eqref{eq:var_aplusb} gives the claimed result.
\end{proof}

\thmprivacy*
\begin{proof}
  Fix a user $j$ and two items $v_j, v_j' \in \mathcal{V}$ held by
  $j$. Observe that, in Algorithm~\ref{alg:client_side_variant}, for
  any fixed sequence $\mathcal{H}$ of hash functions each user $j$
  makes a report to $KT$ channels, and each report is generated
  independently. Among $K$ channels, there exists only one channel on
  which user $j$ sends the noisy report of her true item $v_j$. On the
  remaining $K-1$, user $j$ sends the noisy report of a special item
  $\vec{0}$. Thus, changing the user's item from $v_j$ to $v_j'$
  changes the distribution of user's report on at most $2T$
  channels, and on each channel the ratio of two distributions is
  bounded by $\exp(\frac{\varepsilon_{\mathsf{HH}}}{2T})$ by the $\varepsilon$-LDP
  property of the extended randomizer. Since user's reports over
  separate channels are independent, the corresponding ratio over all
  the $KT$ channels are bounded by $\exp(\frac{2T\varepsilon_{\mathsf{HH}}}{2T})=\exp(\varepsilon_{\mathsf{HH}})$. For
  frequency oracle, user $j$ generates another report using
  $\mathsf{OLH}$, which satisfies $\varepsilon_{\mathsf{OLH}}$-LDP, and
  sends it to the server. Again, by independence of user's reports for
  heavy hitter detection and frequency oracle, the ratio of user's
  output distribution is bounded by
  $\exp(\varepsilon_{\mathsf{HH}}) \cdot \exp(\varepsilon_{\mathsf{OLH}})=\exp(\varepsilon_{\mathsf{HH}} +
\varepsilon_{\mathsf{OLH}})$. This completes the proof. 
\end{proof}

\proplowvar*
\begin{proof}
Using the parameters in~\eqref{eq:ext_rand_param}, we get the variance
of frequency estimate $\hat{f}(v)$ given by the extended randomizer:  
\begin{footnotesize}
\begin{align}
  \Var(\hat{f}(v))
  &= \frac{1}{n}\bigl\{
    f(v)\cdot (c^2(p+q)-1) + (1 - f(v)) \cdot 2c^2\theta\bigr\}
    \nonumber \\
  &=\frac{1}{n}\left\{
    f(v)\cdot
    \left(\frac{3(e^\varepsilon - 1)}{(e^{\varepsilon}-1)^2}\right) + 
    \frac{2(e^{\varepsilon}+2)}{(e^{\varepsilon}-1)^2}\right\}\,.
    \label{eq:var_ext_appdx}
\end{align}
\end{footnotesize}
The variance of $\hat{f}(v)$ for the basic randomizer is
\begin{footnotesize}
\begin{align}
  \Var(\hat{f}(v))
  &=
    \frac{1}{n}\left\{\left(\frac{e^{\varepsilon}+1}{e^{\varepsilon}-1}\right)^2
    - f(v)\right\}\,.
    \label{eq:var_bas_appdx}
\end{align}
\end{footnotesize}
To find the values of $\varepsilon$ such that
\eqref{eq:var_ext_appdx}$\leq$ \eqref{eq:var_bas_appdx}, we set
\[
  f(v)
  \left(\frac{3(e^\varepsilon - 1)}{(e^{\varepsilon}-1)^2}\right) + 
  \frac{2(e^{\varepsilon}+2)}{(e^{\varepsilon}-1)^2}
  \leq \left(\frac{e^{\varepsilon} +1}{e^{\varepsilon}-1}\right)^2 -
  f(v)\,.
\]
Simplifying and rearranging the terms, the above inequality reduces to
\begin{footnotesize}
\begin{equation}
  (f(v)-1)e^{2\varepsilon} + f(v)e^{\varepsilon} + (3-2f(v))
  \leq 0\,. \label{eq:quad_ineq}
\end{equation}
\end{footnotesize}
\begin{figure}[tp]
	\centering
	\includegraphics[scale=0.7]{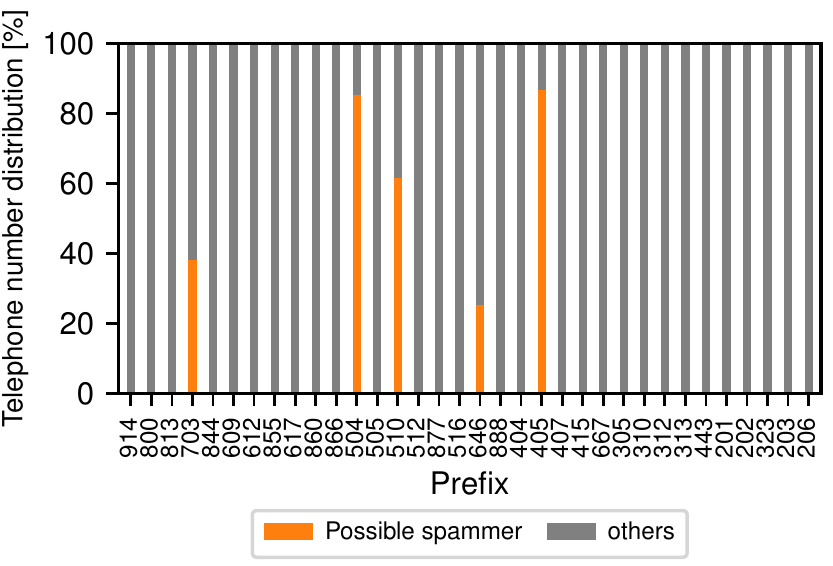}
	\vspace{-5pt}
	\caption{Telephone number distribution of a sample day. Striped bars are related to phone numbers that received more than 100 complaints.}
	\label{fig:bucket_distro}
\end{figure}
Substituting $t = e^{\varepsilon}$ and $a=f(v)$, we see that the
l.h.s. term of the above inequality is a simple quadratic function
$g(t) = (a-1)t^2 + at + (3-2a)$, where 
$0\leq a < 1$. The quadratic function $g$ is concave and has zeros
at
{\small
\[
  t = \frac{a \pm \sqrt{a^2 - 4(a-1)(3-2a)}}{1 - a}\,.
\]
}
Thus, the inequality~\eqref{eq:quad_ineq} is satisfied when
{\small
\[
  \varepsilon \geq \ln \frac{a + \sqrt{9a^2-20a + 12}}{1-a}\,.
\]
}
\end{proof}
\begin{figure}[tp]
	\centering
	\includegraphics[scale=0.7]{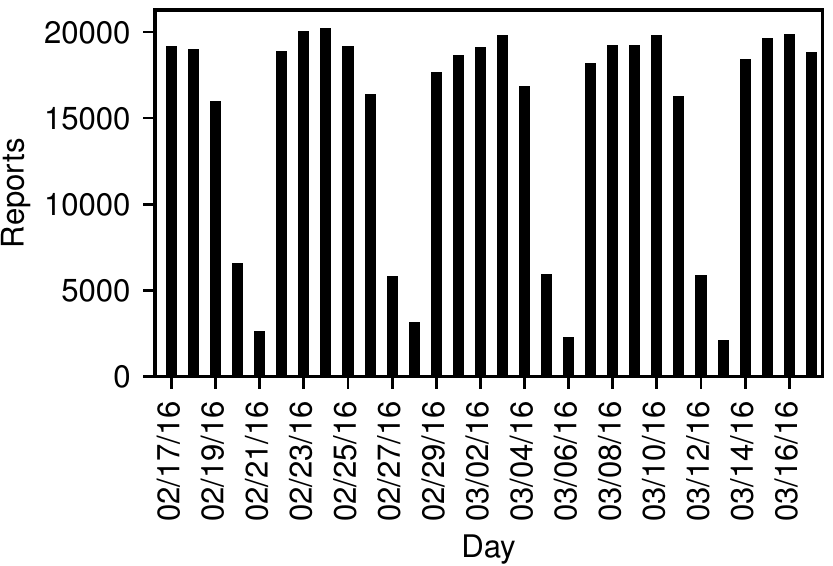}
	\vspace{-5pt}
	\caption{Number of daily complaints received between Feb. 17th and Mar. 17th.}
	\label{fig:ftc_dataset_stats}
\end{figure}

\section{Dataset Properties}
\label{apx:dataset}

Figure~\ref{fig:bucket_distro} shows the relative frequency of phone numbers that make more than one hundred calls in a day, compared to the total number of calls made by all phone numbers reported within the same area code. These graphs are computed based on phone numbers extracted from unwanted call reports from US residents to the FTC (more details about the FTC data we use are provided in Section~\ref{sec:evaluation}).
Each vertical bar indicates a different area code prefix. The striped portion of
the bars indicates the relative fraction of complaints related to numbers that
were complained about more than one hundred times in a day. The figure is
related to a sample day worth of reports. As can be seen, phone numbers with
more than one hundred complaints appear only in a limited number of prefixes.
Their relative occurrence frequency is high in their respective area codes,
whereas it would be diluted if we considered all 10-digit numbers in just one
bucket.

Figure~\ref{fig:ftc_dataset_stats} depicts the number of valid reports received
each day, showing a weekly pattern in which a much lower number of complaints is
received around the weekends. Figure~\ref{fig:complaints_distribution} shows the
distribution of the number of complaints per caller ID. Specifically, the
$x$-axis lists the number of complaints, and the $y$-axis show how many phone
numbers have received $x$ complaints in a single day, throughout the entire
period of observation included in the dataset. It is easy to see that the vast
majority of phone numbers received a single daily complaint, but there also
exist many phone numbers that received hundreds of complaints in a single day.

\begin{figure}[h!]
	\centering
	\includegraphics[scale=0.7]{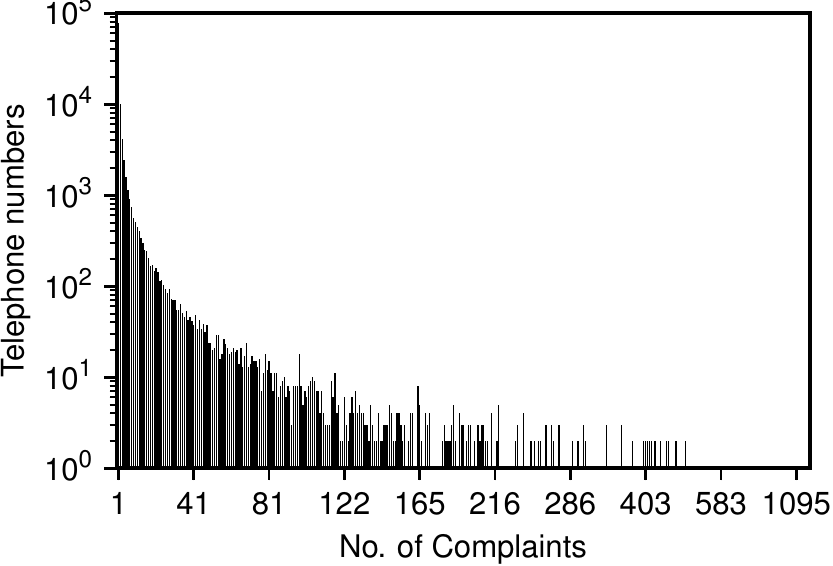}
	\vspace{-5pt}
	\caption{Distribution of daily complaints per caller ID.}
	\label{fig:complaints_distribution}
\end{figure}
	
\end{document}